\documentclass[11pt,a4paper]{article}
\usepackage{datatool}
\usepackage{etoolbox}

\DTLnewdb{bibnotes}

\usepackage{url}
\makeatletter
\patchcmd{\@lbibitem}%
   {\item[\@biblabel {#1}\hfill ]}%
   {
     \item[\@biblabel {#1}\hfill ]%
     \DTLforeach[\DTLiseq{\mylabel}{#2}]{bibnotes}{\mylabel=mylabel,\mynote=mynote}{\textit{\mynote}}
   }{}{\message{^^JPatching failed^^J}}%
 \patchcmd{\@bibitem}%
   {\item}%
   {
     \item%
     \DTLforeach[\DTLiseq{\mylabel}{#1}]{bibnotes}{\mylabel=mylabel,\mynote=mynote}{\textit{\mynote}}
   }{}{\message{^^JPatching failed^^J}}%

\usepackage{amsmath}
\usepackage{amsthm}
\usepackage{amssymb}
\usepackage{answers}
\usepackage{setspace}
\usepackage{graphicx}
\usepackage{enumitem}
\usepackage{multicol}
\usepackage{blkarray}
\usepackage{tikz}
\usetikzlibrary{automata, positioning, arrows}
\usepackage[colorlinks = true,
            linkcolor = blue,
            urlcolor  = blue,
            citecolor = blue,
            anchorcolor = blue,pagebackref=true]{hyperref}
\renewcommand*\backref[1]{\ifx#1\relax \else Cited on page(s): #1. \fi}
\usepackage{mathrsfs}
\usepackage{lmodern}
\usepackage{amsmath,amsthm,amssymb}
\usepackage[T1]{fontenc}
\usepackage[noend]{algpseudocode}
\usepackage{comment}
\usepackage{textcomp}
\usepackage[ruled,vlined,linesnumbered]{algorithm2e}
\let\oldnl\nl% Store \nl in \oldnl
\newcommand{\nonl}{\renewcommand{\nl}{\let\nl\oldnl}}% Remove line number for one line
\makeatother

\newcommand{\N}{\mathbb{N}}

\newcommand{\C}{\mathbb{C}}
\newcommand{\R}{\mathbb{R}}

\newtheorem{theorem}{Theorem}[section]
\newtheorem{corollary}{Corollary}[theorem]
\newtheorem{lemma}[theorem]{Lemma}

\newtheorem{definition}[theorem]{Definition}

\usepackage[utf8]{inputenc}
\usepackage[english]{babel}
\usepackage{graphicx}
\graphicspath{{images/}{../images/}}
\usepackage{blindtext}

\usepackage{subfiles}
\title{Subfile Example}
\author{Team Learn Overleaf}
\date{\today}

\begin{document}
\title{Undercomplete Decomposition of Symmetric Tensors in Linear Time, and Smoothed Analysis of the Condition Number}

\author{Pascal Koiran and Subhayan Saha\thanks{P.K. is with Univ Lyon, EnsL, UCBL, CNRS, LIP.
  Email: firstname.lastname@ens-lyon.fr. S.S. is with Department of Mathematics and Operational Research, 
University of Mons, 
Mons, Belgium, Email: firstname.lastname@umons.ac.be}}
\maketitle
\begin{abstract}
{We study symmetric tensor decompositions, i.e., decompositions of the form 
$T = \sum_{i=1}^r u_i^{\otimes 3}$
where $T$ is a symmetric tensor of order 3 and $u_i \in \C^n$.  
In order to obtain efficient decomposition algorithms, it is necessary to 
require additional properties from the $u_i$.
In this paper we assume that the $u_i$ are linearly independent. This implies $r \leq n$, i.e., the decomposition of $T$ is {\em undercomplete}.

We give a randomized algorithm for the following problem in the exact arithmetic model of computation: 
Let $T$ be an order-$3$ symmetric tensor that has an \textit{undercomplete decomposition}.
Then given some $T'$ \textit{close} to $T$,
an accuracy parameter $\varepsilon$, and an upper bound $B$ on the \textit{condition number} of the tensor, output vectors $u'_i$ such that $||u_i - u'_i|| \leq \varepsilon$ (up to permutation and multiplication by {cube roots of unity)} with high probability. The main novel features of our algorithm are:
    \begin{itemize}
        \item We provide the first algorithm for this problem that runs in linear time in the size of the input tensor. More specifically, it requires $O(n^3)$ arithmetic operations for all accuracy parameters $\varepsilon = \frac{1}{\text{poly}(n)}$ and $B = \text{poly}(n)$.
        \item Our algorithm is \textit{robust}, that is, it can handle inverse-quasi-polynomial noise (in $n,B,\frac{1}{\varepsilon}$) { in the input tensor}.
        \item  
        We present a smoothed analysis of the condition number of the tensor decomposition problem. This guarantees that the condition number is low{ with high probability}
        and further shows that our algorithm runs in linear time, except for some rare badly conditioned inputs.
        \end{itemize}
{ Our main algorithm is a reduction to the complete case ($r=n)$ treated in our previous 
work~\cite{KoiranSaha22}. For efficiency reasons we cannot use this algorithm as a blackbox. Instead, we show that it can be run on an {\em implicitly represented} tensor obtained from the input tensor by a change of basis.}
}
\end{abstract}

\section{Introduction}
\subsection{Symmetric tensor decomposition}\label{sec:symmtensordecomp}
{ Let $T \in \C^n \otimes \C^n \otimes \C^n$ be a symmetric tensor of order $3$. We recall that such an object can be viewed
as a 3-dimensional array $(T_{ijk})_{1 \leq i,j,k \leq n}$ that 
is invariant under all 6 permutations of the indices $i,j,k$.
This is therefore a 3-dimensional generalization of the notion 
of  symmetric matrix.}
In this paper, we { study}
symmetric tensor decompositions, i.e., decompositions of the form 
\begin{equation}\label{eq:decompdef}
    T = \sum_{i=1}^r u_i \otimes u_i \otimes u_i
\end{equation}
where $u_i \in \C^n$. The smallest possible value of $r$ is the symmetric tensor rank of $T$ and it is NP-hard to compute already for $d=3$.
This was shown by Shitov~\cite{Shi16}, and a similar NP-hardness result for ordinary tensors was obtained much earlier by H{\aa}stad~\cite{Hstad'89}.
In this paper, we impose an additional linear independence condition on the $u_i$. Note that such a decomposition is unique if it exists, up to a permutation of the $u_i$'s and scaling by cube roots of unity~\cite{KRUSKAL77,Har70}. More formally, let $T$ be an order-$3$ symmetric tensor such that 
\begin{equation}\label{eq:rdiagonalisable}
    T = \sum_{i=1}^r u_i \otimes u_i \otimes u_i \text{ where } u_i \text{ are linearly independent.}
\end{equation}
We will call such a tensor $r$-\textit{diagonalisable}.
\par
There is a traditional distinction
between {\em undercomplete} decompositions, 
where $r\leq n$ in~(\ref{eq:decompdef}), and {\em overcomplete} decompositions, where $r>n$.
In this paper, we consider only undercomplete decompositions because of the linear independence condition on the $u_i$.
The case where $r$ is exactly equal to $n$ (referred to as \textit{diagonalisable tensors}) was studied in \cite{KoiranSaha22}.
\par
One can also study the \textit{decision} version of the problem: Given an arbitrary symmetric tensor $T$ and $r \in \N$, is it $r$-diagonalisable ? A randomized polynomial-time algorithm is known for this problem in the algebraic (BSS) model of computation \cite{KS21,KoiranSaha23}.
\subsection{Approximate undercomplete tensor decomposition}\label{sec:fwddef}

{ As explained above,} an order-$3$ symmetric tensor $T \in \C^n \otimes \C^n \otimes \C^n$ is called $r$-diagonalisable if there exist $r$ linearly independent vectors $u_i \in \C^n$ such that $T = \sum_{i=1}^r u_i^{\otimes 3}$. The objective of the $\varepsilon$-approximation problem for tensor decomposition is to find linearly independent vectors $u'_1,...,u'_r$ such that there exists a permutation $\pi \in S_r$ where
\begin{align*}
    ||\omega_iu_{\pi(i)} - u'_i|| \leq \varepsilon
\end{align*}
with $\omega_i$  a cube root of unity. Here $\varepsilon$ is the desired accuracy parameter given as input. 
{Hence the problem is essentially that of approximating the vectors $u_i$ appearing in the
decomposition of $T$. Note that this is a {\em forward approximation} in the sense of numerical analysis. 
\subsection{Results and Techniques}
\label{sec:results}

Recall that an order-$3$ tensor $T \in (\C^n)^{\otimes 3}$ is called diagonalisable if there exist linearly independent vectors $u_1,...,u_n \in \C^n$ such that $T$ can be decomposed as in (\ref{eq:decompdef}).
\begin{definition}[Condition number of an $r$-diagonalisable symmetric tensor]
\label{def:conditionnumberundercomplete}
Let $T$ be an $r$-diagonalisable symmetric tensor over $\mathbb{C}$ such that $T = \sum_{i=1}^r u_i^{\otimes 3}$. Let $U \in M_{r,n}(\mathbb{C})$ be the matrix with rows $u_1,\ldots,u_r$. We define the tensor decomposition condition number of $T$ as: $\kappa(T) = ||U||^2_F + ||U^{\dagger}||^2_F$. 
\par
Note here that $||.||_F$ is the Frobenius norm and $U^{\dagger}$ is the Moore-Penrose inverse (refer to Definition \ref{def:pseudoinverse}).
\end{definition}
Note that in the special case of diagonalisable tensors, that is, when $U$ is invertible, $U^{\dagger}$ is equal to $U^{-1}$ in the above expression.
\par
We will show in Section \ref{sec:undercompleterobust} that $\kappa(T)$ is well defined: for a diagonalisable tensor the condition number is independent of the choice of $U$.
Note that when $U$ is close to a singular matrix, the
corresponding tensor is poorly conditioned, i.e., has a large condition number. This is not surprising since 
our goal is to find a decomposition where the vectors $u_i$ are linearly independent.} 
\par
{\textbf{Model of Computation:} The algorithms in this paper are run in the exact arithmetic model of computation. That is, it is assumed that all arithmetic operations over $\C$ can be done exactly. A formalization of the exact arithmetic model can be found in \cite{BSS89,BCSS}}. 
In addition to arithmetic operations, we also allow the computation of square roots 
and cube roots as in~\cite{9317903,KoiranSaha22}.
As in~\cite{KoiranSaha22}, we need cube roots to compute
certain scaling factors (see for instance step~7 of Algorithm~\ref{algo:Jennrich}).\footnote{We will assume that a complex number is stored as a pair of real numbers (its real and imaginary parts). Our model is therefore closer to BSS over $\mathbb{R}$ than over $\mathbb{C}$. This is necessary because 
we need to compute complex conjugates (see for instance step 2 of Algorithm~\ref{algo:robustSUBrecovery} or Step 1 of Algorithm~\ref{algo:undercompleteexactalgebraic}).
Moreover, the diagonalization algorithm from~\cite{9317903} relies on the
QR factorization from~\cite{DDH07}, which requires the
computation of the 2-norm of complex vectors. This is intrinsically a real number (rather than complex-algebraic) computation.}

\par
Our main result is a \textit{robust} randomized linear time algorithm for $\varepsilon$-approximate tensor decomposition in the exact arithmetic
model of computation for $\varepsilon = \frac{1}{\text{poly}(n)}$. More formally, the algorithm takes as input a tensor which is \textit{close} to an $r$-diagonalisable tensor, an estimate $B$ for the condition number of the tensor and an accuracy parameter $\varepsilon$ and returns an $\varepsilon$-forward approximate solution to the $r$-undercomplete tensor decomposition problem (following the definition in Section \ref{sec:fwddef}).
\begin{theorem}\label{thm:undercompleteexactprooflinear}
Let $T \in \C^n \otimes \C^n \otimes \C^n$ be an $r$-diagonalisable tensor for some $r \leq n$ and let $T' \in \C^n \otimes \C^n \otimes \C^n$ be such that $||T - T'|| \leq \delta \in~(0, \frac{1}{(nB)^C r^{C\log^4(\frac{rB}{\varepsilon})}})$ for some constant $C$.
Then on input $T'$, a desired accuracy parameter $\varepsilon$ and some estimate $B \geq \kappa(T)$, Algorithm \ref{algo:undercompleteexactlinear} outputs an $\varepsilon$-approximate solution to the tensor decomposition problem for $T$ with probability at least
$$\Big(1 - \frac{13}{r^2}\Big)^2\Big(1 - (\frac{5}{4r} + \frac{1}{4C^2_{CW}r^{\frac{3}{2}}})\Big)$$ where $C_{CW}$ is some constant.
The algorithm requires $O(n^3 + T_{MM}(r)\log^2(\frac{rB}{\varepsilon}))$ arithmetic operations.
\end{theorem}
Here we denote by $T_{MM}(n)$ the number of arithmetic operations required to multiply two $n \times n$ matrices in a numerically stable manner. If $\omega$ denotes the exponent of matrix multiplication, it is known that 
$T_{MM}(n) = O(n^{\omega+\eta})$ 
 for all $\eta > 0$ (see~\cite{9317903,DDHK'07} for details).

 Very roughly, the algorithm of Theorem~\ref{thm:undercompleteexactprooflinear} 
 consists of an optimized reduction to the case $r=n$ treated
 in~\cite{KoiranSaha22}. This is explained in more detail in Sections~\ref{sec:undercompletealgorithm} and~\ref{sec:overview}. In particular, the bound 
 in Theorem~\ref{thm:undercompleteexactprooflinear} on the noise $\delta$ that can be tolerated comes from~\cite{KoiranSaha22}, and
 can be ultimately traced back to~\cite{9317903}.
  {{ We do not claim any optimality for the
   tolerance to noise stated in Theorem~\ref{thm:undercompleteexactprooflinear}: the upper bound on $\delta$ could probably be improved with a more careful analysis.
 Our algorithm can be viewed as an optimized version
 of the well-known "simultaneous diagonalisation" or "Jennrich" algorithm~\cite{moitra18}.  }
 
In~\cite{KoiranSaha22} we gave a similar result for the diagonalisable
case $r=n$ in the finite precision model. One could also
give a version of Theorem~\ref{thm:undercompleteexactprooflinear}
in finite precision; in order to keep the present paper
within reasonable bounds we will stick to exact arithmetic
in what follows. Note also that the PhD thesis~\cite{Saha23} 
presents a streamlined finite precision analysis compared to~\cite{KoiranSaha22}. The main idea is that numerically stable subroutines can be composed in a numerically stable way
under certain conditions (see also~\cite{BNV23}).

Theorem~\ref{thm:undercompleteexactprooflinear} is in sharp contrast with~\cite{beltran19}, where they give a numerical
instability result for undercomplete tensor decomposition.
Their negative result applies to a wide class 
of "pencil based algorithms" that are fairly close to the algorithm behind Theorem~\ref{thm:undercompleteexactprooflinear}.
The main reason why we can obtain a positive result is that our algorithm is randomized, whereas~\cite{beltran19} only considers deterministic algorithms (see sections~1.4 and~1.5 of~\cite{KoiranSaha22} for a more thorough discussion).

When the estimate $B$ on $\kappa(T)$ is reasonably small (for instance, polynomial in $n$) the bound on the
number of arithmetic operations in Theorem~\ref{thm:undercompleteexactprooflinear} will be dominated by the $O(n^3)$ term. This justifies the "linear time"
claim in the title of the paper since there are $\Theta(n^3)$ entries in a symmetric tensor of size $n$.
As a complement to Theorem~\ref{thm:undercompleteexactprooflinear},
we provide a smoothed analysis of the condition number showing
that it is indeed reasonable to expect $\kappa(T)$ to be quite
small. A more detailed discussion of this result can be found
at the end of the Introduction, and the technical developments 
are in Section~\ref{sec:smoothedanalysis}. This result is already interesting in the case of complete decompositions  
($r=n$) since~\cite{KoiranSaha22} did not include any average-case or smoothed analysis of the condition number.
{ Prior work on smoothed analysis of tensor decompositions can be found in e.g.~\cite{BCMV14}.  The main focus of that paper is on overcomplete decomposition of higher-order tensors. For undercomplete decomposition of order 3 tensors, they also provide an analysis of the "simultaneous diagonalization" or "Jennrich" algorithm. In contrast with the
present paper they do not propose a notion of condition 
number for tensor decomposition, and  do not provide explicit exponents for the running time of their algorithms. }

\subsubsection{Outline of the Algorithm}\label{sec:undercompletealgorithm}

{   As mentioned before, we proceed by reduction
to the problem of complete decomposition treated in~\cite{KoiranSaha22}, 
where $r=n$ in~(\ref{eq:rdiagonalisable}). 
For this we determine the span of 
the $u_i$ in~(\ref{eq:rdiagonalisable}), and after a change of basis
we just have to perform a complete decomposition in $r$-dimensional space. At a high level, this is the strategy suggested in~\cite{kayal11} to decompose homogeneous polynomials in 
sums of powers of linear forms.\footnote{This is indeed the same problem as symmetric tensor decomposition, see for instance~\cite{KS21}.} More precisely, in the terminology 
of~\cite{kayal11} determining the span
of the $u_i$ amounts to finding the essential variables of the
input polynomial. In order to obtain a {robust} algorithm
with the running time claimed in Theorem~\ref{thm:undercompleteexactprooflinear} we need a number
of additional ingredients, which we present in the remainder of Section~\ref{sec:undercompletealgorithm} and in Section~\ref{sec:overview}. 

\par
The "change of basis" operation mentioned in the previous paragraph applies a linear map of the form $A \otimes A \otimes A$ to the input tensor.} Here, $A \in~M_{n,r}(\C)$ and we apply
$A$ to the 3 components of the tensor. In particular,
for rank-1 symmetric tensors in $ \C^n \otimes \C^n \otimes \C^n$, we have 
\begin{equation}\label{eq:changeofbasisdefinition}
    (A \otimes A \otimes A).(u \otimes u \otimes u)=(A^Tu)^{\otimes 3} \in  \C^r \otimes \C^r \otimes \C^r.
\end{equation} 
We give more details on this operation at the beginning of Section~\ref{sec:cobfinitear}.
\par
Before describing the algorithm, we recall that an order-$3$ symmetric tensor $T \in~\C^n \otimes \C^n \otimes \C^n$ can be cut into $n$ slices $T_1,\ldots,T_n$ where $T_k = (T_{ijk})_{1 \leq i,j \leq n}$. Each slice is a symmetric matrix of size~$n$.
The algorithm proceeds as follows:
\begin{enumerate}
    \item[(i)] Pick vector $a = (a_1,...,a_n)$ at random from a finite set and compute a random linear combination $T^{(a)} = \sum_{i=1}^n a_i T_i$ of the slices $T_1,...,T_n$ of~$T$.
    \item[(ii)] Compute a compact singular value decomposition  $T^{(a)} = P\Sigma Q^*$ of this matrix (refer to Definition \ref{def:svd}).
    \item[(iii)]  Let $T' = (\overline{P} \otimes \overline{P} \otimes \overline{P}).T$, {  where  $\overline{P}$ is the complex conjugate of $P$.}   Using the algorithm for the diagonalisable case, compute linearly independent vectors $u_1,...,u_r \in \C^r$ such that $T' =~\sum_{i=1}^r u_i^{\otimes 3}$.
    \item[(iv)] Output $l_1,...,l_n$ where $l_i = Pu_i$ for all $i \in [n]$.
\end{enumerate}
In Section \ref{sec:algalgoundercomplete}, we describe our algorithm in greater detail and show that if it is given an $r$-diagonalisable tensor exactly as input, it indeed outputs a (unique) decomposition. 
We show in in Section \ref{sec:undercompleterobust} that
the algorithm is robust to errors in the input tensor.
{In order to achieve the desired running time,
it turns out that (contrary to what might be expected from the above simplified presentation) we cannot afford to run the algorithm from \cite{KoiranSaha22} on $T'$ in a black box way at step~(iii). We expand on this point and on other important ingredients in Section~\ref{sec:overview}. }
\subsubsection{{  The algorithm in more detail}}\label{sec:inmoredetail}
\label{sec:overview}
An $n \times r$ matrix $A$ (where $n \geq r$) is called semi-unitary if $A^*A = I_r$.
\par
\textbf{Semi-Unitary Basis Recovery Problem.} The input to this problem is an $r$-diagonalisable tensor $T \in \C^n \otimes \C^n \otimes \C^n$, i.e., a tensor which can be decomposed
like in~(\ref{eq:rdiagonalisable}) as $T =~\sum_{i=1}^r u_i^{\otimes 3}$ where $u_1,\ldots,u_r \in \C^n$ are linearly independent. The goal of the \textit{semi-unitary basis (SUB) recovery problem} is to find a semi-unitary matrix $P$ such that its column span is exactly equal to $\text{span}\{u_1,...,u_r\}$. We show in Section~\ref{sec:SUBrecovery} that if $P$ is a solution to the SUB recovery problem, then the tensor
\begin{equation}\label{eq:changeofbasisT'}
T' = (\overline{P} \otimes \overline{P} \otimes \overline{P}).T  \in \C^r \otimes \C^r \otimes \C^r  
\end{equation}
is indeed $r$-diagonalisable. 
Moreover, we show in Section \ref{sec:undercompleterobust} that {  the condition numbers of $T$ and $T'$ are equal.
This relies crucially on the fact that $P$ is semi-unitary.}
If there is an additional assumption that the condition number of $T$ is bounded,  the tensor decomposition algorithm for diagonalisable tensors from \cite{KoiranSaha22} can then be used to compute a decomposition for $T'$. Moreover, we show that a decomposition for $T'$ can be reused to compute a decomposition for $T$ as well. 
\par
We give a randomized algorithm for a \textit{robust version} of this problem. More formally, we show that given some tensor $T'$ which is close to an $r$-diagonalisable tensor $T$, some desired accuracy parameter $\varepsilon$ and some estimate $B \geq \kappa(T)$, our algorithm (Algorithm \ref{algo:robustSUBrecovery}) outputs a solution which is at  distance at most $\varepsilon$ from a solution to the SUB recovery problem for $T$ with high probability. As discussed in Section \ref{sec:undercompletealgorithm} (Step (iii) of the outline), our approach 
relies on the computation of a rank-$r$ singular value decomposition (SVD) of a random linear combination of the slices of the tensor $T$. To implement a \textit{robust} version of this algorithm, we use internally the DEFLATE algorithm from \cite{9317903}. Let $A$ be the desired input matrix and let $S$ be an orthonormal matrix such that its columns span the range 
of~$A$. Then given some input $\Tilde{A}$ \textit{close} to $A$ and some desired accuracy parameter~$\eta$, the DEFLATE algorithm computes some matrix $\Tilde{S}$ which is {  at distance at most $\eta$} away from $S$ (in the operator norm) in matrix multiplication time with high probability. 
\par
One crucial technical step involved here is that the {probability of error} of the DEFLATE algorithm on desired input $A$ is a function of $\frac{1}{\sigma_{min}(A)}$ where $\sigma_{min}(A)$ is the smallest non-zero singular value of $A$.
By step (i) of the outline in Section \ref{sec:undercompletealgorithm}, we apply the DEFLATE algorithm on $T^{(a)}$. We show {  in Lemma~\ref{lem:singvalue}} that the smallest non-zero singular value of $T^{(a)}$ is bounded sufficiently far away from $0$ and this gives us the desired error bounds.
\par
\textbf{Linear Time Implementation.}
The previous discussion is not enough to give a linear time algorithm for  undercomplete tensor decomposition. This is because if $P \in \C^{n \times r}$ is a solution to the SUB recovery problem for the tensor $T \in \C^n \otimes \C^n \otimes \C^n$, then an 
explicit computation of the tensor $T' = (\overline{P} \otimes \overline{P} \otimes \overline{P}).T$ as defined in (\ref{eq:changeofbasisT'})
cannot be performed in time $O(n^3)$  {  to the best of our knowledge}. In Appendix \ref{appsec:cobexact}, we give an algorithm for computing $T'$
based on fast rectangular matrix multiplication \cite{GU'18} {and we show that it runs in $O(n^{3.251})$ arithmetic operations. This is better than the naive algorithm,
but this bound  falls short of the $O(n^3)$ complexity that would be 
needed for a linear time algorithm.}
\par
In order to achieve linear time bounds (in the input size), the algorithm needs to avoid computing all the entries of $T'$ explicitly. {  At step (iii) of the algorithm's outline, we will therefore run our algorithm for complete tensor decomposition~\cite{KoiranSaha22} on an input tensor $T'$ which is {\em implicitly described} 
by equation~(\ref{eq:changeofbasisT'}). For this we need to "open the box" of the algorithm from~\cite{KoiranSaha22}.
This algorithm relies in particular on the computation 
of two random linear combinations of the slices of the input tensor. In Section \ref{sec:cobfinitear}, we show how
to do this when the input $T'$ is implicitly described 
by~(\ref{eq:changeofbasisT'}).
This is based on an auxiliary algorithm from \cite{KoiranSaha22} for computing a linear combination of the slices of a tensor after a change of basis by a square matrix, which we extend to the case of rectangular matrices.} 
Then we show in Section \ref{sec:undercompletelineartime}  how the algorithm for the complete case \cite{KoiranSaha22} can be modified to output a decomposition of $T'$ using access to these two random linear combinations of the slices of $T'$ and the desired input tensor $T$.
\par

\par
\textbf{Condition Numbers: }In Section \ref{sec:smoothedanalysis}, we perform a smoothed analysis of the condition number for undercomplete tensor decomposition (as in Definition \ref{def:conditionnumberundercomplete}). More formally, let $T$ be an arbitrary symmetric tensor of rank $r \leq n$.
Then $$T = (U \otimes U \otimes U).(\sum_{i=1}^r e_i^{\otimes 3})=\sum_{i=1}^r u_i^{\otimes 3}$$ where the $u_i$ are the columns of $U$.
Let $U'$ be a matrix sampled from the space of matrices obtained by random Gaussian perturbations of the entries of $U$ 
and $T' = (U' \otimes U' \otimes U').(\sum_{i=1}^r e_i^{\otimes 3})$. Or equivalently, define  $T'$ as $\sum_{i=1}^r {u'_i}^{\otimes 3}$ where the entries of $u'_i$ are obtained from those of $u_i$ by random Gaussian perturbations. Then $\kappa(T')$ is at most $\text{poly}(n,||U||, \frac{1}{\sigma},\sigma)$ with high probability, where $\sigma^2$ is the variance of the Gaussian distribution. This proof relies on certain technical lemmas about the analysis of condition numbers from \cite{BC13}. As explained in the paragraph before Section~\ref{sec:undercompletealgorithm}, this analysis shows that our main algorithm runs in linear time, except for some rare badly conditioned inputs. Currently, our analysis applies only to real tensors and real perturbations. A similar analysis in the complex setting remains to be done.
A more elementary proof of the same statement when $U$ is picked at random from a centered normal distribution on $n \times n$ matrices can also be found 
in Appendix~\ref{app:simpleavgcasediag}.

{\subsection{Organization of the paper}

In Section~\ref{sec:prelim} we recall some results from
our previous paper on complete decomposition of symmetric tensors~\cite{KoiranSaha22}, the definition of the Moore-Penrose inverse and some of its properties.
{In particular, we use these properties in Section~\ref{sec:welldefined} to show that the condition number of the input tensor is well-defined.}
In Section~\ref{sec:cobfinitear} we present
some algebraic algorithms related to the change of basis operation (as explained in Section~\ref{sec:inmoredetail}, they
play a crucial role in the derivation of a decomposition
algorithm running in linear time).
Section~\ref{sec:SUBrecovery} is devoted to the semi-unitary basis recovery problem (see also Section~\ref{sec:inmoredetail}).
Then we combine all these ingredients in Section~\ref{sec:algoundercomplete} to obtain a fast
and robust algorithm for undercomplete tensor decomposition.
Finally, the smoothed analysis of the condition number 
is carried out in Section~\ref{sec:smoothedanalysis}.
}

\section{Preliminaries}
\label{sec:prelim}
\subsection{Norms and Condition Numbers: }
We denote by $||x||$ the $\ell^2$ (Hermitian) norm of a vector $x \in \C^n$. For $A \in M_n(\C)$, we denote by $||A||$ its operator norm and by $||A||_F$ its Frobenius norm:
$$||A||^2_F=\sum_{i,j=1}^n |A_{ij}|^2.$$
We always have $||A|| \leq ||A||_F$. 
\newline
\begin{definition}[Tensor Norm]\label{def:tensornorm}
Given tensor $T \in (\C^n)^{\otimes 3}$, we define the Frobenius norm $||T||_F$ of $T$ as
\begin{align*}
    ||T||_F = \sqrt{\sum_{i,j,k=1}^n |T_{i,j,k}|^2}
\end{align*}
\end{definition}
If $T_1,...,T_n$ are the slices of $T$, we also have that
\begin{equation}\label{eq:tensornormslices}
 \sum_{i=1}^n||T_i||^2_F = \sum_{i,j,k \in [n]} |
(T_i)_{j,k}|^2 = ||T||_F^2. 
\end{equation}

For a given invertible matrix $V$, we define the \textit{Frobenius condition number} to be
\begin{equation}\label{eq:Frobeniusconditionnumber}
  \kappa_F(V) = ||V||_F^2 + ||V^{-1}||_F^2  .
\end{equation}
\begin{definition}[Condition number of a diagonalisable symmetric tensor]
\label{def:conditionnumber}
Let $T$ be a diagonalisable symmetric tensor over $\mathbb{C}$: 
{ we have $T = \sum_{i=1}^n u_i^{\otimes 3}$
where the $u_i$ are linearly independent vectors.}
Let $U \in M_n(\mathbb{C})$ be the matrix with rows $u_1,\ldots,u_n$. We define the tensor decomposition condition number of $T$ as: $\kappa(T) {= \kappa_F(U)} = ||U||^2_F + ||U^{-1}||^2_F$.
\end{definition}
{ We have shown in~\cite{KoiranSaha22} that $\kappa(T)$ is independent of the choice of~$U$. This is a special case of Definition~\ref{def:conditionnumberundercomplete}, which applies more generally to $r$-diagonalisable tensors.
We will show in Section~\ref{sec:welldefined} that the condition number of Definition~\ref{def:conditionnumberundercomplete} is also independent of $U$.}

\subsection{Algebraic algorithm from \cite{KoiranSaha22}}

The following is the algorithm for computing a decomposition for a diagonalisable tensor from \cite{KoiranSaha22}. In this algorithm, the inherent assumption is that one can compute the eigenvectors exactly. { In Section~\ref{sec:finite} we will refine this idealized algorithm into a more realistic algorithm, also from~\cite{KoiranSaha22}, where the eigenvectors are computed approximately with the diagonalisation algorithm from~\cite{9317903}.}

\begin{algorithm}[H] \label{algo:completeexact}
\SetAlgoLined
\nonl \textbf{Input:} An order-3 diagonalisable symmetric tensor $T \in \C^{n \times n \times n}$.  \\
\nonl \textbf{Output:} linearly independent vectors $l_1,...,l_r \in  \C^n$ such that $T = \sum_{i=1}^n l_i^{\otimes 3}$ \\
\nonl Pick $a_1,...,a_n$ and $b_1,...,b_n$ uniformly and independently from a finite set $S \subset \C$ \\
\nonl Let $T_1,...,T_n$ be the slices of $T$ \\
Compute $T^{(a)} = \sum_{i=1}^n a_iT_i$ and $T^{(b)} = \sum_{i=1}^n b_iT_i$  \\
Compute $T^{(a)'} = (T^{(a)})^{-1}$ \\
Compute $D = T^{(a)'}T^{(b)}$ \\
Compute the normalized eigenvectors $p_1,...,p_n$ of $D$. \\
Let $P$ be the matrix with $(p_1,...,p_n)$ as columns and compute $P^{-1}$. Let $v_i$ be the $i$-th row of $P^{-1}$ \\
Define $S = (P \otimes P \otimes P).T$ and let $S_1,...,S_n$ be the slices of $S$. Compute $\alpha_i = \text{Tr}(S_i)$. \\
Output $(\alpha_1)^{\frac{1}{3}}v_1,...,(\alpha_n)^{\frac{1}{3}}v_n$
\caption{ Complete decomposition of symmetric tensors.}
\end{algorithm}The following theorem from \cite{KoiranSaha22} shows that the above algorithm actually returns a solution to the tensor decomposition problem for a diagonalisable tensor $T$.
\begin{theorem}\label{thm:Jennrichexact}
Given a diagonalisable tensor $T \in \C^n \otimes \C^n \otimes \C^n$, Algorithm~\ref{algo:completeexact} returns $u_1,...,u_n \in \C^n$ such that $T = \sum_{i=1}^n u_i^{\otimes 3}$ with probability at least $1- \frac{1}{n}$.  
\end{theorem}
{\subsubsection{Explanation for Step 6 of the algorithm}\label{sec:explanationstep6}
After step 3 of the algorithm, the algorithm has already determined vectors $v_1,\ldots,v_n$ such that
\begin{equation}\label{eq:inexplanationeq1}
T=\sum_{i=1}^n \alpha_i v_i^{\otimes 3}.    
\end{equation}
The goal of the rest of the algorithm is to find the unknown coefficients
$\alpha_i$. It is shown in \cite{KoiranSaha22}  that the system can be solved quickly by exploiting some of its structural properties. The approach relies on a change of basis defined by~(\ref{eq:changeofbasisdefinition}).
 Let $V$ be the matrix with rows $v_1,...,v_n$ and then following Algorithm \ref{algo:completeexact}, let $V = P^{-1}$. Then (\ref{eq:inexplanationeq1}) can also be rewritten as 
 $$T = (V \otimes V \otimes V).(\sum_{i=1}^n \alpha_i e_i^{\otimes 3})$$ where $e_i$ are the standard basis vector in $\C^n$. 
Then $$T' = (P \otimes P \otimes P). T = (VP \otimes VP \otimes VP).(\sum_{i=1}^n \alpha_i e_i^{\otimes 3}) = (\sum_{i=1}^n \alpha_i e_i^{\otimes 3}).$$ From this one can conclude that $\alpha_i = Tr(T'_i)$ where $T'_1,...,T'_n$ are the slices of $T'$.}
\subsection{Algorithm in finite precision from \cite{KoiranSaha22}} \label{sec:finite}
The following is the algorithm from \cite{KoiranSaha22} for computing a tensor decomposition for a diagonalisable tensor in finite precision arithmetic. { It is a  
version of the well-known "simultaneous diagonalisation" or "Jennrich" algorithm, optimized for diagonalisable symmetric tensors. The main goal of the present paper is to extend this algorithm from the diagonalisable case, 
where $r=n$ in~(\ref{eq:rdiagonalisable}), to $r<n$. See for instance~\cite{KoiranSaha22,moitra18} for additional background on the traditional version of Jennrich's algorithm.}

Let $\eta \in (0,1)$ such that $\frac{1}{\eta}$ is an integer. Define the discrete grid
\begin{equation}\label{eq:discretegrid}
    G_{\eta} = \{-1,-1+\eta,-1+2\eta,...,1-2\eta,1-\eta\}.
\end{equation}

\begin{algorithm}[H] \label{algo:Jennrich}
\SetAlgoLined
\nonl In the following algorithm, ${ C} ,C_{\text{gap}},C_{\eta} > 0$ and $c_F > 1$ are some absolute constants that are set in \cite{KoiranSaha22}. \\ 
\nonl \textbf{Input:} An order-3 symmetric diagonalisable tensor $T \in (\mathbb{C}^n)^{\otimes 3}$, an estimate $B$ for the condition number of the tensor and  an accuracy parameter $\varepsilon (< 1)$.  \\
\nonl \textbf{Output:} A solution to the $\varepsilon$-forward approximation problem { (in the sense of Section~\ref{sec:fwddef})} for decomposition of the tensor $T$. \\
{ \nonl Pick $(a_1,...,a_n,b_1,...,b_n) \in G_{\eta}^{2n}$ uniformly at random where $\eta := \frac{1}{C_{\eta}n^{\frac{17}{2}}B^4}$ is the grid size. \\}
Compute $T^{(a)} = \sum_{i=1}^n a_iS_i$ and $T^{(b)} = \sum_{i=1}^n b_iT_i$. \\
Compute $T^{(a)'} = (T^{(a)})^{-1}$ \\
Compute $D = T^{(a)'}T^{(b)}$ \\
\nonl Set $k_{\text{gap}} := \frac{1}{C_{\text{gap}}n^6B^3}$, $k_{F} := c_Fn^5B^3$ and { $\varepsilon_1 := \frac{\varepsilon^3}{Cn^{12}B^{\frac{9}{2}}}$.} \\
Let $v_1,...,v_n$ be the output of EIG-FWD on the input $(D,{\varepsilon_1}, \frac{3nB}{k_{\text{gap}}}, 2B^{\frac{3}{2}}\sqrt{nk_F})$ where EIG-FWD is the numerically stable algorithm for approximate matrix diagonalisation from \cite{9317903},\cite{KoiranSaha22}.\\
Compute $W = V^{-1}$ and let $w_1,...,w_n$ be the rows of $W$. \\
Compute $\alpha_1,...,\alpha_n = TSCB(V,T)$ where TSCB is the algorithm for computing the trace of slices of a tensor after change of basis as described in Algorithm \ref{algo:fastcob}. \\
Compute $z_i = \alpha_i^{\frac{1}{3}}w_i$ for all $ i \in [n]$. \\
\nonl Output $z_1,...,z_n$.
\caption{ Approximate decomposition of diagonalisable tensors~\cite{KoiranSaha22}.}

\end{algorithm}
{ In \cite{KoiranSaha22} the underlying computation model is finite precision arithmetic with adversarial error, i.e., the only assumption on the result an arithmetic computation is that it is within some specified error of the
exact result. For this reason, this algorithm is a fortiori 
valid in the exact arithmetic model used in the present paper.}

{At line 4 of this algorithm, the  second input $\varepsilon_1$ to  EIG-FWD is an accuracy parameter: 
EIG-FWD will output the normalized eigenvectors of $D$
with error at most $\varepsilon_1$ in $\ell^2$ norm for each eigenvector.
The last two inputs to  EIG-FWD are estimates of the condition number of the eigenproblem and of the Frobenius norm of the matrix $D$ to be diagonalized. In order to understand the present paper, it
is not particularly important to know what these parameters are and how the estimates at line 4 are derived (but the interested reader will find all details in~\cite{9317903,KoiranSaha22}). What matters is the following result from~\cite{KoiranSaha22}:}
\begin{theorem}\label{thm:completetensordecomposition}
Let $T \in \C^{n} \otimes \C^n \otimes \C^n$ be a diagonalisable tensor such that $\kappa(T) \leq B$, { where $\kappa(T)$ is the condition number of Definition~\ref{def:conditionnumber}.} 
Given some $T' \in \C^{n} \otimes \C^n \otimes \C^n$ such that $||T - T'||~\leq \delta~\in ~\Big(0, \frac{1}{n^{c\log^{4}(\frac{nB}{\varepsilon})}}\Big)$ for some constant $c$, Algorithm \ref{algo:Jennrich} outputs an $\varepsilon$-forward approximate solution to the tensor decomposition problem for $T$ in 
$$O(n^3 + T_{MM}(n)\log^2 \frac{nB}{\varepsilon})$$ arithmetic operations on a floating point machine 
with probability at least $\Big(1- \frac{1}{n} - \frac{12}{n^2}\Big)\Big(1 - \frac{1}{\sqrt{2n}} - \frac{1}{n}\Big)$.
\end{theorem}
{ In~\cite{KoiranSaha22}, we moreover show that the above computation can be carried out with a polylogarithmic number of bits of precision. More precisely, we show that $-\log \delta$ bits suffice, where $\delta$ is as in Theorem~\ref{thm:completetensordecomposition}. In particular, it is only assumed that the input is stored with precision $\delta$.
This is the reason why the algorithm can tolerate a perturbed
input $T'$ in the above theorem.}
\subsection{Moore-Penrose Inverse}\label{sec:pseudoinverse}

 In this section we recall the definition and some basic properties of the Moore-Penrose inverse.
\begin{definition}\label{def:pseudoinverse}
    For a matrix $A \in K^{m \times n}$, a \textit{pseudoinverse} of $A$ is defined as a matrix $A^{\dagger} \in K^{n \times m}$ satisfying all of the following three criteria, known as the Moore-Penrose conditions:
    \begin{enumerate}
        \item $AA^{\dagger}$ maps all column vectors of $A$ to themselves, that is, $AA^{\dagger}A = A$.
        \item $A^{\dagger}$ acts as a weak inverse, that is, $A^{\dagger}AA^{\dagger} = A^{\dagger}$.
        \item $AA^{\dagger}$ and $A^{\dagger}A$ are Hermitian matrices.
    \end{enumerate}
\end{definition}
\textbf{Some properties of the Moore-Penrose inverse:}
\begin{enumerate}
    \item The Moore-Penrose inverse is unique for all matrices over $\R$ and $\C$.
    \item The following are sufficient conditions for $(AB)^{\dagger} = B^{\dagger}A^{\dagger}$:
    \begin{enumerate}
        \item $A$ has orthonormal columns or $B$ has orthonormal rows or
        \item $A$ has linearly independent columns (then $A^{\dagger}A = I$) and $B$ has linearly independent rows (then $BB^{\dagger} = I$) or
        \item $B = A^{*}$ or $B = A^{\dagger}$
    \end{enumerate}
    \item Let $M \in \text{M}_r(\mathbb{K})$ be an invertible matrix and let $\Tilde{M} = \begin{bmatrix}
        &M \\
        &0_{(n-r) \times r}
    \end{bmatrix} \in~\C^{n \times r}$  be constructed by adding rows of $0$ to the matrix. Then the pseudoinverse of $\Tilde{M}$ is given by $\Tilde{M}^{\dagger} \in \C^{r \times n} = \begin{bmatrix}
        M^{-1} & 0_{ r \times (n-r)}
    \end{bmatrix}$
    \begin{proof}
    The first two properties are standard, so 
    we only give the (simple) proof of the third property.
    
    We first want to show that $\Tilde{M}^{\dagger}  = \begin{bmatrix}
        M^{-1} & 0_{ r \times (n-r)}
    \end{bmatrix}$ indeed satisfies the properties of the definition of the Moore-Penrose inverse (Definition~\ref{def:pseudoinverse}). By the properties of block matrix multiplication, we get that
    \begin{align*}
    \Tilde{M}(\Tilde{M})^{\dagger} &= \begin{bmatrix}
        &M \\
        &0_{(n-r) \times r}
    \end{bmatrix}\begin{bmatrix}
        M^{-1} & 0_{ r \times (n-r)}
    \end{bmatrix}   = \begin{bmatrix}
        I_r & 0 \\
        0 & 0 
    \end{bmatrix} \\
        \Tilde{M}^{\dagger} \Tilde{M}  &= \begin{bmatrix}
        M^{-1} & 0_{ r \times (n-r)}
    \end{bmatrix} \begin{bmatrix}
        &M \\
        &0_{(n-r) \times r}
    \end{bmatrix} = I_r + 0_{r \times r} = I_r.
    \end{align*}
    Hence both $\Tilde{M}^{\dagger} \Tilde{M}$ and $\Tilde{M}(\Tilde{M})^{\dagger}$ are Hermitian matrices and this shows that point (3) in Definition \ref{def:pseudoinverse} is satisfied.
    \par
    Moreover since, 
    \begin{align*}
        \Tilde{M}(\Tilde{M})^{\dagger}\Tilde{M} = \begin{bmatrix}
        I_r & 0 \\
        0 & 0 
    \end{bmatrix}  \begin{bmatrix}
        &M \\
        &0_{(n-r) \times r}
        \end{bmatrix} = \Tilde{M}
    \end{align*}
    and 
    \begin{align*}
        \Tilde{M}^{\dagger} \Tilde{M} \Tilde{M}^{\dagger}  = I_r \Tilde{M}^{\dagger}  = \Tilde{M}^{\dagger} 
    \end{align*}
    this shows that the first two properties of Definition \ref{def:pseudoinverse} are satisfied as well. The result then follows from the uniqueness of the Moore-Penrose inverse.
    \end{proof}
\end{enumerate}

{\subsection{The condition number is well defined}}
\label{sec:welldefined}

The Frobenius condition number of a square invertible matrix was defined in (\ref{eq:Frobeniusconditionnumber}). We extend it to rectangular matrices as follows:
\begin{equation}\label{eq:Frobeniuscondnumrect}
    \kappa_F(U) = ||U||_F^2 + ||U^{\dagger}||_F^2,
\end{equation}
where $U^{\dagger}$ is the Moore-Penrose inverse of $U$.
Recall that we had already defined the condition numbers for undercomplete tensor decomposition in Definition \ref{def:conditionnumberundercomplete}.  We state it here again for the convenience of the reader (note that this is a generalization of Definition \ref{def:conditionnumber}).
\newline
\textbf{Definition }\ref{def:conditionnumberundercomplete}: Let $T \in \mathbb{C}^n\otimes \C^n \otimes \C^n$ be an $r$-diagonalisable tensor such that $T = \sum_{i=1}^r u_i^{\otimes 3}$ where the vectors $u_i \in \mathbb{C}^n$ are linearly independent and $r \leq n$. Let $U \in \C^{r \times n}$ be such that the $u_i$ are the rows of $U$. We say that "$U$ diagonalises $T$" and we define the condition number of $T$, 
\begin{align*}
 \kappa(T) := \kappa_F(U) = ||U||_F^2 + ||U^{\dagger}||_F^2
\end{align*}
where $U^{\dagger}$ is the Moore-Penrose inverse of $U$. 

The following lemma from \cite{kayal11} shows that if a tensor is $r$-diagonalisable, then it has a unique decomposition up to permutations and multiplication by cube roots of unity
(it is actually stated in \cite{kayal11} as a result about sums of powers of linear forms). Similar uniqueness results hold for ordinary tensors \cite{Har70},  under  more general conditions than linear independence of the components \cite{KRUSKAL77}
{(and these two uniqueness results in fact imply Lemma~\ref{lem:uniqueness}).}
\begin{lemma}\label{lem:uniqueness}\cite{kayal11}
Let $T = \sum_{i \in [n]} u_i^{\otimes 3}$ where $u_i$ are linearly independent vectors over $\mathbb{C}$. For any other decomposition $T = \sum_{i \in [n]} (u'_i)^{\otimes 3}$, the vectors $u'_i$ must satisfy $u'_i = \omega_iu_{\pi(i)}$ where $\omega_i$ is the cube root of unity and $\pi \in S_n$ a permutation. 
\end{lemma}
We'll use the above lemma to show that the condition number for $r$-diagonalisable tensors in Definition~\ref{def:conditionnumberundercomplete} is well-defined.
\begin{lemma}
Let $T$ be a $r$-diagonalisable tensor. Then for all matrices $U \in M_{r,n}(\mathbb{C})$ that  diagonalise $T$, the quantities $||U||_F^2 + ||U^{\dagger}||_F^2$ are equal.
\end{lemma}
\begin{proof}
By Lemma \ref{lem:uniqueness},  for all $U \in M_{r,n}(\mathbb{C})$ such that $U$ $r$-diagonalises $T$, the rows of $U$ are unique up to permutation and scaling by cube roots of unity. Writing this in matrix notation, if $U$ and $U'$ are two such distinct matrices that diagonalise the tensor $T$, there exists a permutation $\pi \in S_r$ and a diagonal matrix $D$ with cube roots of unity along the diagonal entries, such that $U' = DP_{\pi}U$ where $P_{\pi}$ is the permutation matrix corresponding to $\pi$. 
\par 
{The Frobenius norm of a matrix is invariant under permutation of
its rows, or multiplications by cube roots of unity.
Hence $||U||_F^2 = ||U'||_F^2$,} and it remains to show that $||U^{\dagger}||_F^2 = ||(U')^{\dagger}||_F^2$. Since the rows of $U$ are linearly independent and the columns of $DP_{\pi}$ are linearly independent, $$||(U')^{\dagger}||_F = ||(DP_{\pi}U)^{\dagger}||_F = ||U^{\dagger}P_{\pi}^{\dagger}D^{\dagger}||_F$$
{by Section~\ref{sec:pseudoinverse}, property 2.b.}
Since $P_{\pi}$ is a permutation matrix,  its columns are orthonormal. Hence $(P_{\pi})^{\dagger} = (P_{\pi})^{-1} = (P_{\pi})^T$ and  multiplication by $(P_{\pi})^{\dagger}$ on the right permutes the columns of $U^{\dagger}$. Also, inverse of cube roots of unity are cube roots of unity as well. Hence, if $v'_1,...,v'_n$ are the columns of $(U')^{\dagger}$, and $v_1,...,v_n$ are the columns of $U^{\dagger}$, this gives us that $v'_i = \omega'_iv_{\pi^{-1}(i)}$ where $\omega'_i$ are cube roots of infinity. This gives us that $||(U')^{\dagger}||^2_F = \sum_{i=1}^r ||v'_i||^2 = \sum_{i \in [r]} ||\omega'_iv_{\pi^{-1}(i)}||^2 = \sum_{i \in [n]} ||v_i||^2 = ||U^{\dagger}||_F^2$. This finally gives us that for all $U,U' \in M_{r,n}(\mathbb{C})$ such that $U$ and $U'$ diagonalise $T$, $||U'||_F^2 + ||(U')^{\dagger}||_F^2 = ||U||_F^2 + ||U^{\dagger}||_F^2$.
\end{proof}
\section{Change of Basis}\label{sec:cobfinitear}

Given tensors $T,T' \in \mathbb{C}^{n \times n \times n}$, we say that there is a change of basis $A \in~\text{M}_r(\mathbb{C})$ for some $r  \leq n$ that takes $T$ to $T' \in (\C^r)^{\otimes 3}$ if $T' = (A \otimes A \otimes A).T$. Written in standard basis notation, this corresponds to the fact that for all $i_1,i_2,i_3 \in [r]$,
\begin{equation}\label{eq:changeofbasisdef}
    T'_{i_1i_2i_3} = \sum_{j_1,j_2,j_3 \in [n]} A_{j_1i_1}A_{j_2i_2}A_{j_3i_3} T_{j_1j_2j_3}.
\end{equation}
Note that if $T = u^{\otimes 3}$ for some vector $u \in \C^n$, then $(A \otimes A \otimes A).T = (A^Tu)^{\otimes 3}$.
{ From this one can obtain the effect of two successive changes of basis on an arbitrary symmetric tensor $T$:
\begin{equation} \label{eq:2changes}
 (B \otimes B \otimes B).(A \otimes A \otimes A).T = (AB \otimes AB \otimes AB).T,
\end{equation}
a fact which will be used in Section~\ref{sec:undercompletelineartime}.}

 We do not know how to perform a change of basis in
time $O(n^3)$ (the best we can do is presented in Appendix~\ref{appsec:cobexact}). For this reason, as discussed in Section~\ref{sec:inmoredetail} we present in Sections~\ref{sec:tscb} and~\ref{sec:lcscb} two "implicit algorithms" related to this operation.
{The first one is an extension to rectangular matrices of a similar algorithm from~\cite{KoiranSaha22} while
the second one is a new ingredient.}
\par
The following theorem describes the structure of the slices of a tensor after a change of basis by a rectangular matrix.
\begin{theorem}\label{thm:P3structural}
Let $T \in \C^n \otimes \C^n \otimes \C^n$ be a tensor with slices $T_1,...,T_n$ and let $S = (V \otimes V \otimes V).T$ where $V \in M_{n,r}(\C)$ for some $r \leq n$. Then the slices $S_1,...,S_r$ of $S$ are given by the formula:
\begin{align*}
    S_k = V^T D_k V
\end{align*}
 where $D_k = \sum_{i=1}^n v_{i,k} T_i$ and $v_{i,k}$ are the entries of $V$.
\par
In particular, if $T=\sum_{i=1}^n e_i^{\otimes 3}$, we have $D_k = \text{diag}(a_{1,k},...,a_{n,k} )$. 
\end{theorem}
\begin{proof}
Using the definition of the change of basis operation, we get that
\begin{equation}
    (S_{i_1})_{i_2,i_3} = \sum_{j_1,j_2,j_3 = 1}^n V_{j_1,i_1}V_{j_2,i_2}V_{j_3,i_3}T_{j_1,j_2,j_3} = \sum_{j_2,j_3 = 1}^n V_{j_2,i_2} (\sum_{j_1 = 1}^n V_{j_1,i_1}T_{j_1,j_2,j_3}) V_{j_3,i_3}
\end{equation}
Writing this in matrix form gives us that
\begin{align*}
    S_{i_1} = V^T \Big(\sum_{j_1 = 1}^n V_{j_1,i_1}T_{j_1}\Big) V
\end{align*}
which gives us the desired result.
\end{proof}
\begin{corollary}\label{corr:P3structural}
Let $S =\sum_{i=1}^r a_i^{\otimes 3}$. Let $A$ be the $r \times n$ matrix with rows $a_1,...,a_r$. Then the slices $S_k$ of $S$ are given by the formula 
$$S_k = A^T D_k A \text{ where } D_k = \text{diag}(a_{1,k},...,a_{r,k}).$$
\end{corollary}

\subsection{Trace of Slices after a Change of Basis}
\label{sec:tscb}

Recall the definition of the change of basis operation from (\ref{eq:changeofbasisdef}). In this section, we give an algorithm which takes in an order-$3$ symmetric tensor $T \in \C^n \otimes \C^n \otimes \C^n$ and some rectangular matrix $V \in \C^{n \times r}$ and computes the trace of the slices $T'_1,...,T'_r$ of the tensor $T' = (V \otimes V \otimes V).T \in \C^r \otimes \C^r \otimes \C^r$.
A similar algorithm was already proposed in \cite{KoiranSaha22} but in the case where the change of basis matrix $V$ was a square matrix.

\begin{algorithm}[H] \label{algo:fastcob}
\SetAlgoLined
\nonl \textbf{Input:} An order-$3$ symmetric tensor $T \in \mathbb{C}^{n \times n \times n}$, a matrix $V = (v_{ij}) \in \mathbb{C}^{n \times r}$.\\
\nonl Let $T_1,...,T_n$ be the slices of $T$. \\
Compute $W = VV^T$  \\
Compute $x_{m,k} = (WT_m)_{k,k}$ for all $m,k \in [n]$. \\
Compute $x_m = \sum_{k=1}^n x_{m,k}$  for all $m \in [n]$ . \\
Compute $s_i = \sum_{m=1}^n v_{m,i}x_m$ for all $i \in [n]$. \\
\nonl Output $s_1,...,s_n$
\caption{Algorithm that outputs the trace of the slices after a change of basis (TSCB)}
\end{algorithm}
\begin{theorem}\label{thm:fastcob}
Let us assume that a tensor $T \in (\C^n)^{\otimes 3}$ and a matrix $V \in M_{n,r}(\C)$ for some $r \leq n$ are given as input to Algorithm \ref{algo:fastcob}. Set $S = (V \otimes V \otimes V).T \in \C^r \otimes \C^r \otimes \C^r$ following the definition in (\ref{eq:changeofbasisdef}) and let $S_1,...,S_r \in M_r(\C)$ be the slices of $S$. Then the algorithm returns $s_1,...,s_r$ where $s_i = Tr(S_i)$ for all $i \in [r]$ using  $O(n^3)$ operations. 
\end{theorem}
\begin{proof}
  We first claim that $Tr(S_i) = \sum_{m=1}^n v_{mi}\Big(\sum_{k=1}^n (VV^TT_m)_{k,k}\Big)$. Using Theorem \ref{thm:P3structural}, we know that $S_i = V^TD_iV$ where $D_i = \sum_{m=1}^n V_{m,i}T_m$. Using the cyclic property and the linearity of the trace operator, we get that
  \begin{align*}
      \text{Tr}(S_i) = \text{Tr}(V^TD_iV) = \text{Tr}(VV^TD_i) &= \text{Tr}(VV^T (\sum_{m=1}^n v_{m,i}T_m)) \\
      &= \sum_{m=1}^n v_{m,i}\text{Tr}(VV^T (\sum_{m=1}^n v_{m,i}T_m)) \\
      &=  \sum_{m=1}^n v_{m,i}\Big(\sum_{k=1}^n (VV^TT_m)_{k,k}\Big).
  \end{align*}
From this, we conclude that Algorithm \ref{algo:fastcob} computes
exactly the trace of the slices $S_i$ of $S$ .
\par
\textbf{Running Time:} We analyse the steps of the algorithm and deduce the number of arithmetic operations required to perform the algorithm. Note that only the numbered steps contribute to the complexity analysis. 
\begin{enumerate}
    \item Since $V \in M_n(\C)$, Step 1 can be done { in $O(r^2n)$ operations with ordinary matrix multiplication.}
    \item In Step 2, for each $m,k \in [n]$, we compute the inner product of the $k$-th row of $W$  with the $k$-th column of $T_m$. Computation of each inner product takes $n$ arithmetic operations. There are $n^2$ such inner product computations. So this step requires $n^3$ arithmetic operations. 
    \item In Step 3, we compute each $x_m$ by adding $x_{m,k}$ for all $k \in [n]$. Thus each $x_m$ requires $n$ arithmetic operations and hence, this step requires $n^2$ arithmetic operations. 
    \item In Step 4, we compute each $\Tilde{s}_i$ by taking the inner product of the $i$-th column of $V$ and $X = (X_1,...,x_m)$. 
    Each inner product requires~$n$ arithmetic operations and hence, this step requires $n^2$ arithmetic operations.  
\end{enumerate}
\end{proof}

\subsection{Linear Combination of Slices after a Change of Basis}
\label{sec:lcscb}
In this section, we give an algorithm which takes as input an order-$3$ symmetric tensor $T \in \C^n \otimes \C^n \otimes \C^n$, some rectangular matrix $V \in \C^{n \times r}$ and some vector $a \in \C^r$. The algorithm returns a linear combination $T^{(a)} = \sum_{i=1}^r a_iS_i$ of the slices $S_1,...,S_r$ of the tensor $S = (V \otimes V \otimes V).T \in \C^r \otimes \C^r \otimes \C^r$. 
\par
{ This algorithm plays a crucial role in the linear time algorithm for undercomplete tensor decomposition (Algorithm \ref{algo:undercompleteexactlinear}). As discussed at the beginning of Section \ref{sec:cobfinitear}, we do not know how to compute the explicit tensor after a change of basis in linear time. We show later in Section \ref{sec:undercompletelineartime} that we do not need to compute the entire tensor - just computing two random linear combinations of the slices of the tensor is enough.}

\begin{algorithm}[H] \label{algo:changeofvariableslincomb}
\SetAlgoLined
\nonl \textbf{Input:} An order-$3$ symmetric tensor $T \in \mathbb{C}^{n \times n \times n}$, a matrix $V = (v_{ij}) \in \mathbb{C}^{n \times r}$, a vector $a \in \C^r$.\\
\nonl Let $v_1,...,v_n$ be the rows of $V$ \\
Compute $\alpha_i = \langle a,v_i \rangle$  for all $i \in [n]$. \\
\nonl Let $T_1,...,T_n$ be the slices of $T$ \\
Compute matrices $\{X_m\}_{m \in [n]}$ such that $(X_m)_{j,k} = \alpha_m (T_m)_{j,k}$ for all $j,k \in [n]$ . \\
Compute $D^{(a)} = \sum_{m=1}^n X_m$. \\
Compute $A = D^{(a)}V$ \\
Compute $S^{(a)} = V^TA$ \\
\nonl Output $S^{(a)}$.
\caption{Algorithm for computing a linear combination of slices after change of basis ($\text{LCSCB}$) in a symmetric tensor}
\end{algorithm}

\begin{theorem}\label{thm:changeofvarslincomb}
Let us assume that an order-$3$ symmetric tensor $T \in \C^n \otimes \C^n \otimes \C^n$, a matrix $V \in \C^{n \times r}$ and a vector $a \in \C^r$ (where $r \leq n$) 
are given as input to Algorithm \ref{algo:changeofvariableslincomb}. 
Let $S = (V \otimes V \otimes V).T$ where $S_1,...,S_r$ are the slices of $S$.  Then Algorithm \ref{algo:changeofvariableslincomb} computes $S^{(a)} = \sum_{i=1}^r a_iS_i$ using $O(n^3)$ arithmetic operations.
\end{theorem}
\begin{proof}
Since $S_1,...,S_r$ are the slices of the tensor $S = (V \otimes V \otimes V).T$, using Theorem \ref{thm:P3structural}, we have that
\begin{equation}\label{eq:sislices}
    S_i = V^T D_i V
\end{equation}
where $D_i = \sum_{m=1}^n v_{im} T_m$. Let $v_1,...,v_n$ be the rows of $V$. Then, using (\ref{eq:sislices}) we have that 
\begin{align*}
    S^{(a)} = \sum_{i=1}^r a_iS_i = V^T\Big(\sum_{i=1}^r a_iD_i\Big)V &= V^T\Big(\sum_{i=1}^r a_i(\sum_{m=1}^n v_{mi} T_m)\Big)V  \\
    &=  V^T\Big(\sum_{m=1}^n (\sum_{i=1}^r a_iv_{mi}) T_m\Big)V \\
    &=  V^T\Big(\sum_{m=1}^n \langle a, v_m \rangle T_m\Big)V.
\end{align*}
\textbf{Running Time:} We analyse the steps of the algorithm and deduce the number of arithmetic operations required to perform the algorithm. Note that only the numbered steps contribute to the complexity analysis. 
\begin{enumerate}
    \item In Step 1, we need to compute $n$ inner products of vectors in $\C^r$. This can be done in $O(nr)$.
    \item In Step 2, for each $m\in [n]$, we perform scalar multiplication of the matrices $X_m$ by $\alpha_m$. Scalar multiplication of each matrix takes $n^2$ arithmetic operations and hence, the total number of arithmetic operations in this step is $n^3$. 
    \item In Step 3, we compute $D^{(a)}$ by adding $X_{m}$ for all $m \in [n]$. Computing each entry requires $n$ arithmetic operations and since there are $n^2$ entries, this step requires $n^3$ arithmetic operations. 
    \item In Step 4, we multiply an $r \times n$ matrix and an $n \times n$ matrix which can be done in $O(n^2r)$ many arithmetic operations.  
    \item In Step 5, we multiply an $r \times n$ matrix with an $n \times r$ matrix which requires $O(nr^2)$ arithmetic operations.  
\end{enumerate}
So, the total number of arithmetic operations required is $O(n^3)$.

\end{proof}

\section{Semi-Unitary Basis Recovery Problem}\label{sec:SUBrecovery}
{ Recall that we call a symmetric tensor $T \in \mathbb{C}^{n} \otimes \C^n \otimes \C^n$ $r$-diagonalisable if
\begin{equation}\label{eq:r-diagdefagain}
T = \sum_{i=1}^r u_i^{\otimes 3} \text{ where } \{u_i\}_{i \in [r]} \subseteq \mathbb{C}^n \text{ are linearly independent.}    
\end{equation}
}
In this section, we look at the following algorithmic problem.
\begin{definition}[SUB recovery]\label{def:ONBrecovery}
Given an $r$-diagonalisable tensor $T \in \C^n \otimes \C^n \otimes  \C^n$ that can be written as $T = \sum_{i=1}^r u_i^{\otimes 3 }$, find a semi-unitary matrix $n \times r$ matrix $P $ with columns $p_1,...,p_r$ such that $\text{span}\{u_1,...,u_r\} = \text{span}\{p_1,...,p_r\}$.
\end{definition}
\par
{ Solving this algorithmic problem (robustly and in linear time) is a crucial step for our algorithm for the undercomplete tensor decomposition problem. As explained in Section \ref{sec:undercompletealgorithm}, the strategy for the undercomplete tensor decomposition algorithm is similar to that of \cite{kayal11} - find a basis for the subspace spanned by the vectors of the decomposition ($u_1,...,u_r$ in (\ref{eq:r-diagdefagain})) and then use that to reduce it to the algorithm for the complete case. Here, for reasons mentioned in Sections \ref{sec:inmoredetail} and Sections \ref{sec:undercompleterobust} (after Definition \ref{def:undercompleteproblemdef}), in fact, we  want to compute a \textit{semi-unitary basis} (the elements of the basis form a semi-unitary matrix) of the subspace spanned by the vectors of the decomposition.}
\par
In this section, we give an algorithm for this problem. To describe the algorithm, we first need to recall the notion of singular value decomposition of matrices.
\begin{definition}\label{def:svd}
For any matrix $A \in \C^{m \times n}$, the singular value decomposition is defined to be a factorization of the form $M = U\Sigma V^*$, where $U \in \C^{m \times m}$ and $V \in \C^{n\times n}$ are unitary matrices and $\Sigma$ is an $m\times n$ rectangular diagonal matrix with non-negative real numbers on the diagonal.
\par
Let us assume without loss of generality that the singular values in $\Sigma$ are sorted in decreasing order of magnitude. That is, if $\Sigma = \text{diag}(\sigma_1,...,\sigma_{\min\{m,n\}})$, then $|\sigma_i| \geq |\sigma_j|$.
\par
Let $\text{rank}(A) = r \leq \min\{m,n\}$. Then the compact singular value decomposition is defined to be a factorization of the form $M = U_r\Sigma_rV_r^*$ where $\Sigma_r$ is an $r \times r$ diagonal matrix with only positive real numbers on the diagonal and $U_r \in \C^{m \times r}$ and ${ V_r^*} \in \C^{r\times n}$ are semi-unitary matrices i.e. $U_r^*U_r = V_rV_r^* = I_r$.
\end{definition}

\subsection{Algebraic algorithm for SUB recovery}
{ We first give an algorithm for computing an exact solution to the SUB recovery problem for an $r$-diagonalisable tensor $T$ (which is given exactly), assuming that the singular value decomposition of a matrix can be computed exactly.}
{ In  Section~\ref{sec:robustsub} we will refine
this idealized algorithm into a robust approximate algorithm
for SUB recovery.}

\begin{algorithm}[H] \label{algo:ONBrecovery}
\SetAlgoLined
\nonl \textbf{Input:} An order-3 $r$-diagonalisable symmetric tensor $T \in \C^n \otimes \C^n \otimes \C^n$.  \\
\nonl \textbf{Output:} Semi-unitary matrix $P$ which is a solution to the SUB recovery problem for $T$\\
\nonl Pick $\alpha_1,...,\alpha_n$ uniformly and independently from a finite set $S \subset \C$ \\
\nonl Let $T_1,...,T_n$ be the slices of $T$ \\
Set $T^{(\alpha)} = \sum_{i=1}^n \alpha_iT_i$. \\
Compute the compact singular value decomposition of $T^{(\alpha)} = P_r\Sigma Q_r^{*}$  as in Definition \ref{def:svd} \\
\nonl Output $P_r$
\caption{Algorithm for SUB recovery problem}
\end{algorithm}

The following is the main theorem of this section.
\begin{theorem}\label{thm:ONBrecovery}
If the input tensor $T \in \C^n \otimes \C^n \otimes \C^n$ is $r$-diagonalisable, then Algorithm \ref{algo:ONBrecovery} returns a solution to the SUB recovery problem for $T$ (refer to Definition \ref{def:ONBrecovery}) with high probability. { More precisely,} let us assume that the input tensor $T$ can be written as $T = \sum_{i=1}^{ r} u_i^{\otimes 3}$ where the $u_i$ are linearly independent. Then, if $\alpha_1,...,\alpha_n$ are picked uniformly and independently at random from a finite set $S \subset \C$
such that $\text{span}\{p_1,...,p_r\} = \text{span}\{u_1,...,u_r\}$ with probability $1- \frac{r}{|S|}$ .
\end{theorem}

The proof of this theorem follows a two-step process:
\begin{itemize}
    \item First, we show in Lemma \ref{lem:redONB} that if the $\alpha_1,...,\alpha_n \in \C$ picked in Algorithm \ref{algo:ONBrecovery} have some "specific properties", then the Algorithm indeed returns a solution to the SUB recovery problem.
    \item Then, we show that if the $\alpha_i$ are picked uniformly and independently at random from a finite set, then they indeed have those "specific properties" { with high probability}.
\end{itemize}
\begin{lemma}\label{lem:redONB}
Let $T \in (\mathbb{C}^n)^{\otimes 3}$ be an $r$-diagonalisable tensor. Then $T$ can be written as $T = \sum_{i=1}^r u_i^{\otimes 3}$ where the $u_i$'s are linearly independent. Let $T_1,...,T_n \in M_n(\mathbb{C})$ be the slices of $T$ and let $T^{(\alpha)} = \sum_{i=1}^n \alpha_iT_i$ be a linear combination of the slices such that $\text{rank}(T^{(\alpha)}) = r$. Take $ P_r\Sigma Q_r^{*}$ to be a compact singular value decomposition of $T^{(\alpha)} $ as in Definition \ref{def:svd}. Let $p_1,...,p_r$ be the $r$ columns of $P_r$. Then
\begin{align*}
    \text{span}(p_1,...,p_r) =\text{Im}(T^{(\alpha)}) = \text{span}(u_1,...,u_r).
\end{align*}
\end{lemma}
\begin{proof}
Firstly, let $A \in \C^{m \times n}$ be any matrix of rank $r$ and let $A = M\Sigma N^*$ be the compact singular value decomposition of $A$ where $m_1,...,m_r \in \C^m$ are the orthonormal columns of $M$. Then $\text{span}\{m_1,...,m_r\} = \text{Im}(A)$. This follows from the fact that
\begin{align*}
     \text{Im}(A) \subseteq \text{Im}(M) =  \text{span}\{m_1,...,m_r\}.
\end{align*}

Since {the $m_i$'s} are linearly independent, $\text{rank}(A) = r = \text{rank}(\text{span}\{m_1,...,m_r\})$ which gives us the desired conclusion. Applying this for $A = T^{(\alpha)}$ gives us that
\begin{equation}\label{eq:spaneq3}
    \text{Im}(T^{(\alpha)}) = \text{span}\{p_1,...,p_r\}.
\end{equation}
Let $U$ be the $r \times n$ matrix with rows $u_1,...,u_r$. Since $T = \sum_{i=1}^r u_i^{\otimes 3}$, using Corollary~\ref{corr:P3structural}, we get that the slices $T_i$ of $T$ can be written as $U^TD_iU$ where $D_i = \text{diag}(U_{1,i},...,U_{r,i})$. Then $T^{(\alpha)} = U^TD^{(\alpha)}U$ where $D^{(\alpha)} = \text{diag}(\langle \alpha,u_1 \rangle,...,\langle \alpha,u_r \rangle)$ and $\alpha = (\alpha_1,...,\alpha_n)$. This gives us that
\begin{equation}\label{eq:spaneq1}
    Im(T^{(\alpha)}) \subseteq Im(U^T).
\end{equation}
Since $Im(U^T) = \text{span}\{u_1,...,u_r\}$, this gives us that
\begin{equation}\label{eq:spaneq2}
    \text{Im}(T^{(\alpha)}) \subseteq \text{span}\{u_1,...,u_r\}.
\end{equation}
Since the $u_i$'s are linearly independent, $\text{rank}(T^{(\alpha)}) = r = \text{rank}(\text{span}\{u_1,...,u_r\})$. Hence, $\text{Im}(T^{(a)}) = \text{span}\{u_1,...,u_r\}$.
Combining (\ref{eq:spaneq2}) and (\ref{eq:spaneq3}), we get that
$$\text{span}(p_1,...,p_r) = \text{Im}(T^{(\alpha)}) = \text{span}(u_1,...,u_r).$$
\end{proof}
\begin{definition}
The maximal rank of a subspace $\mathcal{U}$ of $M_n(\C)$, denoted by $\text{max-rank}(\mathcal{U})$ is defined as the largest $r$ such that there exists a matrix $M \in \mathcal{U}$ such that $\text{rank}(M) = r$.
\end{definition}
\begin{lemma}\label{lem:undercompcharlemma}
Let $T,T' \in \C^n \otimes \C^n \otimes \C^n$ be such that $T' = (A \otimes A \otimes A).T$ for some invertible matrix $A \in \text{GL}_n(\C)$ . Let $\mathcal{U}$ be the subspace spanned by the slices $T_1,...,T_n$ of $T$ and $\mathcal{V} $ be the subspace spanned by the slices $T'_1,...,T'_n$ of $T$. Then the following properties hold:
\begin{enumerate}
    \item $\mathcal{V} = A^T\mathcal{U}A$.
    \item {$\text{max-rank}(\mathcal{U}) =\text{max-rank}(\mathcal{V})$.}
\end{enumerate}
\end{lemma}
\begin{proof}
Corollary \ref{corr:P3structural} shows us that $V \subseteq A^T\mathcal{U}A$. Since $T = (A^{-1} \otimes A^{-1} \otimes A^{-1}).T$, the same argument shows that $U \subseteq A^{-T}\mathcal{V}A^{-1}$. This gives us that $\mathcal{V} = A^T\mathcal{U}A$. This appears in Lemma 22 of \cite{KS21} but in the language of polynomials.
\par
For the second part of the lemma, {let $r=\text{max-rank}(\mathcal{U})$.}
Then there exists $M_{\mathcal{U}} \in \mathcal{U}$ such that $\text{rank}(M_{\mathcal{U}}) = r$. By the previous part, 
$M_{\mathcal{V}} = A^TM_{\mathcal{U}}A \in \mathcal{V}$. Since $A$ is non-singular, $\text{rank}(M_{\mathcal{V}}) = r$
{ hence $\text{max-rank}(\mathcal{V}) \geq \text{max-rank}(\mathcal{U})$.}
The other direction can be similarly shown by using the fact that $\mathcal{U} = A^{-T}\mathcal{V}A^{-1}$.
\end{proof}
\begin{lemma}\label{lem:maxrankp3}
Let $A_1,...,A_k \in M_{n}(\mathbb{C})$ be such that the subspace they span has maximal rank $r$. Pick $\lambda_1,...,\lambda_k$ uniformly and independently at random from a finite set $S \subseteq \mathbb{C}$. Let $A = \sum_{i \in [k]} \lambda_iA_i$. Then
\begin{align*}
    \text{Pr}_{\lambda_1,...,\lambda_k \in S} [\text{rank}(A) = r] \geq 1- \frac{r}{|S|}
\end{align*}
\end{lemma}
\begin{proof}
Since maximal rank of $\text{span}\{A_1,...,A_k\} = r$, there are some $\mu_1,...,\mu_k$ in $\mathbb{C}$ such that $A_{\mu} = \sum_{i=1}^k \mu_iA_i$ and $\text{rank}(A) = r$. Then there exists an $r \times r$ submatrix of $A_{\mu}$ that is invertible. Let $A_x = \sum_{i=1}^k x_iA_i$ and let $\phi_r(A_x)$ be the determinant of the corresponding $r \times r$ submatrix of $A_x$. Then $\phi_r(A_x)$ is a polynomial in $\mathbb{C}[x_1,..,x_k]$. 
\par
Now $\phi_r(A_x) \not\equiv 0$ Since $\phi_r(A_{\mu}) \neq 0$; Since $\text{deg}(\phi_r(A_x)) \leq r$ , if $\lambda_1,...,\lambda_n$ are picked at random from a finite set $S \subseteq \mathbb{C}$, the probability that $\phi_r(A) \neq 0$ and hence, $rank(A) = r$, is at least $1 - \frac{r}{|S|}$ by the Schwartz-Zippel Lemma.
\end{proof}

\begin{proof}[Proof of Theorem \ref{thm:ONBrecovery}]
Let $E_1$ denote the event that Algorithm \ref{algo:ONBrecovery} returns a solution to the SUB recovery problem for $T$. We want to show that $\text{Pr}_{\alpha_1,...,\alpha_n \in_r S}[E_1]$ is large. 
\par
Let $D = \sum_{i=1}^r (e_i)^{\otimes 3}$ where $e_i$ are the standard basis vectors for $\C^n$. It can be observed that if $\mathcal{U}$ is the span of the slices of $D$, then $\text{max-rank}(\mathcal{U})=~r$. Given $T = \sum_{i=1}^r u_i^{\otimes 3}$, we can write $T = (U \otimes U \otimes U).D$ where $U \in M_n(\C)$ is such that the first $r$ rows of $U$ are $u_1,...,u_r$ and the last $(n-r)$ rows $w_1,...,w_{n-r}$ are picked so that $u_1,...,u_r,w_1,...,w_{n-r}$ form a basis for $\C^n$. This ensures that $U$ is an invertible matrix. Applying Lemma \ref{lem:undercompcharlemma} to $T$ and $D$, we get that if $\mathcal{V}$ is the space of matrices spanned by the slices $T_1,...,T_n$ of $T$, then $\text{max-rank}(\mathcal{V}) = r$. Let $T^{(\alpha)} = \sum_{i=1}^n \alpha_iT_i$ where the $\alpha_i$'s are picked independently and uniformly at random from a finite set $S \subset \mathbb{C}$ and let $E_2$ denote the event that $\text{rank}(T^{(\alpha)}) = r$. Now, using Lemma \ref{lem:maxrankp3} for $T_1,...,T_n$ we get that
\begin{equation}\label{eq:undercompletee2}
    \text{Pr}_{\alpha_1,...,\alpha_n \in S} [E_2] \geq 1-\frac{r}{|S|}.
\end{equation}
Using Lemma \ref{lem:redONB}, we get that if $\text{rank}(T^{(\alpha)}) = r$, the matrix $P_r$ returned by Algorithm \ref{algo:ONBrecovery} is indeed a solution to the SUB recovery problem for $T$. {Thus $E_2 \subseteq E_1$, and the result follows from~(\ref{eq:undercompletee2}).}
\end{proof}
\subsection{Robust SUB recovery}
\label{sec:robustsub}

{ In this section we propose an approximate solution 
to the SUB recovery problem from Definition~\ref{def:ONBrecovery}.}
\begin{definition}[Approximate SUB recovery]\label{def:approxsubrecovery}
Let $T \in \C^n \otimes \C^n \otimes \C^n$ be an order-$3$ $r$-diagonalisable tensor and let $P \in \C^{n \times r}$ be a semi-unitary matrix which is a solution to the SUB recovery problem for $T$. The objective of the $\varepsilon$-approximate SUB recovery problem is to find a matrix $P' \in \C^{n \times r}$ such that $||P-P'|| \leq \varepsilon$.
\end{definition}
{ Our algorithm relies on the $\text{DEFLATE}$ algorithm from~\cite{9317903}. Given some input matrix $A'$ close to a matrix $A$, this algorithm approximately outputs a semi-unitary matrix whose columns span the image space of $A$.
This works for any $A$ belonging to a set of
matrices $X_r(n)$ defined as follows: 
 $$ X_r(n) := \Big\{A \in M_n(\C) \Big| ||A|| \geq \frac{1}{3}, \text{rank}(A) = \text{rank}(A^2) = r\Big\}.$$ 
 More precisely, the following result is established in \cite{9317903}.
\begin{theorem}
\label{thm:deflate}
Let $A \in X_r(n)$ and let $A^{(0)} \in M_n(\C)$ be a matrix such that the following conditions are satisfied:
\begin{enumerate}
    \item $||A - A^{(0)}|| \leq \beta$
    \item $\beta \leq \frac{1}{4} \leq ||A^{(0)}||$.
\end{enumerate}
Given as input  $A^{(0)}$ and an accuracy parameter $\eta$, the $\text{DEFLATE}$ algorithm returns with probability $1 - \frac{(20n)^3\sqrt{\beta}}{\eta^2 \sigma_r(A)}$ a matrix $\Tilde{S} { \in \mathbb{C}^{n \times r} }$ satisfying the following property: there exists a matrix 
$S \in \C^{n \times r}$ with orthonormal columns such that $||S - \Tilde{S} || \leq \eta$, and moreover the columns of $S$  span  $\text{range}(A)$. Here $\sigma_r(A)$ denotes the smallest non-zero singular value of $A$. The number of arithmetic operations required is $O(T_{MM}(n))$.
\end{theorem}
We can now present our algorithm for approximate SUB recovery.}

\begin{algorithm}[H] \label{algo:robustSUBrecovery}
\SetAlgoLined
\nonl Let $T \in \C^n \otimes \C^n \otimes \C^n$ be an order-3 $r$-diagonalisable symmetric tensor. \\
\nonl \textbf{Input:}  A tensor $T' \in \C^n \otimes \C^n \otimes \C^n$ "close" to $T$, an estimate $B$ for the condition number of the tensor and  an accuracy parameter $\varepsilon \leq 1$.  \\
\nonl \textbf{Output:} A solution to the $\varepsilon$-approximate solution to the SUB recovery problem for tensor $T$. \\
\nonl Set $k_F = c_Fn^5B^3$ { where $c_F$ is the absolute constant from Algorithm~\ref{algo:Jennrich}}. \\
\nonl Pick $a \in G^n_{\eta}$ uniformly at random where the grid size $\eta$ is fixed in~(\ref{eq:kappafandeta}). \\
\nonl Let $T_1,...,T_n$ be the slices of $T$ \\
Compute $T^{(a)} = \sum_{i=1}^n a_iT_i$. \\
Compute $A = T^{(a)}(T^{(a)})^{*}$. \\
Compute $\Tilde{A} = 2k_F A$. \\ 
Let $\Tilde{P} = \text{DEFLATE}(\Tilde{A},\varepsilon)$ where $\text{DEFLATE}$ is the numerically stable algorithm for computing orthonormal singular vectors from Theorem \ref{thm:deflate} \\
\nonl Output $\Tilde{P}$.
\caption{Robust numerical algorithm for SUB recovery}
\end{algorithm}
Note that Step 3 of the algorithm is included to ensure that $||\Tilde{A}|| \geq \frac{1}{4}$ so that Theorem \ref{thm:deflate} can be applied.
The following is the main theorem of this section, and we will prove it in Section \ref{sec:finishproofrobustsubrecovery}.
\begin{theorem}\label{thm:SUBrecoverymaintheorem}
Let $T \in \C^n \otimes \C^n \otimes \C^n$ be an $r$-diagonalisable tensor for some $r \leq n$ and let $T' \in \C^n \otimes \C^n \otimes \C^n$ be such that $||T - T'|| \leq \delta \in  \Big[ 0, \frac{\varepsilon^4}{\text{poly}(n,B)}\Big]$~\footnote{The exact bounds for $\delta$ are set in (\ref{eq:delta})}. Then, on input $T'$, a desired accuracy parameter $\varepsilon$ and some estimate $B \geq \kappa(T)$, Algorithm \ref{algo:robustSUBrecovery} outputs an $\varepsilon$-approximate solution to the SUB recovery problem for $T$ with probability at least $(1 - \frac{1}{r})(1 - (\frac{5}{4r} + \frac{1}{4C^2_{CW}r^{\frac{3}{2}}}))$ where $C_{CW}$ is the universal constant from the Carbery-Wright inequality \cite{CW01}. The algorithm requires $O(n^3)$ arithmetic operations.
\end{theorem}
The proof of this theorem has the following two steps:
\begin{enumerate}
    \item We first show in Section \ref{seC:erroranalysis} that if the vector $a \in (-1,1]^n$ chosen in the algorithm has certain ``nice properties" (refer to Definition~\ref{def:inputconditions} for an exact definition), then the algorithm indeed outputs an $\varepsilon$-approximate solution to the SUB recovery problem for the input tensor $T$.
    \item Then we show in Section \ref{sec:probability} that if the vector $a$ is picked uniformly and independently at random from a finite grid (which we describe later),  it will indeed have those ``nice properties" with high probability.
\end{enumerate}
\subsubsection{Error Analysis}\label{seC:erroranalysis}
\begin{definition}[Input Conditions]\label{def:inputconditions}
Let $T \in \C^n \otimes \C^n \otimes \C^n$ be an order-$3$ symmetric tensor and let $a \in (-1,1]^n$. Let $T_1,...,T_n$ be the slices of $T$. We say that $(T,a)$ satisfies the $(r,n,B)$-input conditions with parameter $k_F$ if the following conditions are satisfied
\begin{itemize}
    \item $T$ is an $r$-diagonalisable tensor with $\kappa(T) \leq B$
    \item Let $T^{(a)} = \sum_{i=1}^n a_iT_i$. Then $\text{rank}(T^{(a)}) = r$ and $\kappa_F(T^{(a)}) \leq k_F$ where $\kappa_F(.)$ is the Frobenius condition number of a matrix. 
\end{itemize}  
\end{definition}
The following is the main theorem of Section~\ref{seC:erroranalysis}.
\begin{theorem}\label{thm:erroranalysis}
Let $T \in \C^n \otimes \C^n \otimes \C^n$ be an order-$3$ symmetric tensor and let $a \in (-1,1]^n$ such that $(T,a)$ satisfies the $(r,n,B)$-input conditions with parameter $k_F$ (according to Definition \ref{def:inputconditions}). Let $T' \in (\C^n)^{\otimes 3}$ such that $||T-T'|| \leq \delta \in \Big[ 0, \frac{\varepsilon^4}{24 \times (20)^6 n^7 B^{\frac{3}{2}} k_F}\Big]$. Then Algorithm \ref{algo:robustSUBrecovery} on input $T'$ and accuracy parameter $\varepsilon$ returns a solution to the $\varepsilon$-approximate SUB recovery problem for $T$ with probability at least $1- \frac{1}{r}$.
\end{theorem}
In the remainder of this section, we will refer to $T$ as the desired input to Algorithm \ref{algo:robustSUBrecovery} and to $T'$ as the actual input to the algorithm.
\par
\underline{\textbf{Step 1:}}
On the desired input tensor $T \in \C^n \otimes \C^n \otimes \C^n$ and some vector $a \in (-1,1]^n$, the desired output of \textbf{Step 1} of Algorithm \ref{algo:robustSUBrecovery} is $T^{(a)} = \sum_{i=1}^n a_iT_i$ where $T_1,...,T_n$ are the slices of $T$. Let $T^{(a)'} = \sum_{i=1}^n a_iT'_i$ be the  actual output of Step 1 of the algorithm on actual input $T'$ where $T'_1,...,T'_n$ are the slices of $T'$. Since, $||T - T'|| \leq \delta$, this gives us that
\begin{equation}\label{eq:Step1}
    ||T^{(a)} - T^{(a)'}|| \leq \sum_{i=1}^n |a_i|||T_i - T'_i|| \leq \sqrt{n}||T-T'|| \leq \delta \sqrt{n}.
\end{equation}
\par
\underline{\textbf{Step 2:}}
\begin{lemma}[Lemma 5.16 and (42) in \cite{KoiranSaha22}]\label{lem:normTa}
Let $T_1,...,T_n$ be the slices of a symmetric tensor tensor $T \in \C^n \otimes \C^n \otimes \C^n$. We assume that the input tensor has condition number $\kappa(T) \leq B$. Define $T^{(a)} = \sum_{i=1}^n a_iT_i$ where $a = (a_1,...,a_n) \in [-1,1]^n$. Then $||T^{(a)}|| \leq \sqrt{nB^3}$. 
\end{lemma}

In the error-free situation, \textbf{Step 2} of Algorithm \ref{algo:robustSUBrecovery} takes in as input $T^{(a)}$ from the previous step and computes $A = T^{(a)}(T^{(a)})^*$. Let $A' = T^{(a)'}(T^{(a)'})^*$ be the output of Step 2 of the algorithm on taking the actual input $T^{(a)'}$ from Step 1. Then using (\ref{eq:Step1}), we get that $||(T^{(a)})^* - (T^{(a)'})^*|| \leq \delta \sqrt{n}$. 
From this inequality we have:
\begin{equation}\label{eq:Step2}
\begin{split}
    &||T^{(a)}(T^{(a)})^* - T^{(a)'}(T^{(a)'})^*|| \\
    & \leq ||T^{(a)}(T^{(a)})^* - T^{(a)'}(T^{(a)})^*|| + ||T^{(a)'}(T^{(a)})^* - T^{(a)'}(T^{(a)'})^*|| \\
&\leq \delta \sqrt{n}\Big(||T^{(a)}|| + ||T^{(a)'}|| \Big)    \\
&\leq \delta\sqrt{n}(2\sqrt{nB^3} + \delta \sqrt{n}) \leq 3\delta nB^{\frac{3}{2}}. 
\end{split}
\end{equation}
The second last inequality is an application of the triangle inequality along with Lemma \ref{lem:normTa}. Since, $B$ is an upper bound on the condition number of a matrix, $B \geq 1$. The last inequality stems from the fact that $\delta < 1 < B^{\frac{3}{2}}$.
\par
\underline{\textbf{Step 3:}} Again, in the error-free situation, Step 3 of Algorithm \ref{algo:robustSUBrecovery} ideally takes in as input $A = T^{(a)}(T^{(a)})^*$ from the previous step and computes $\Tilde{A} = 2k_F A$. Let $\Tilde{A'} := 2k_F A'$ be the output of Step 3 of the algorithm on taking as input the actual output $A'$ from Step 2.
Then using (\ref{eq:Step2}), this gives us that 
\begin{equation}\label{eq:Step3}
    ||\Tilde{A} - \Tilde{A'} || \leq 6k_F\delta nB^{\frac{3}{2}}.
\end{equation}
Combining the errors for Steps 1,2 and 3, we obtain the following theorem.
\begin{theorem}\label{thm:errSteps1-3}
Let $T \in \C^n \otimes \C^n \otimes \C^n$ be an order-$3$ symmetric tensor and let $a \in (-1,1]^n$ be such that $(T,a)$ satisfies the $(r,n,B)$-input conditions with parameter $k_F$ (according to Definition \ref{def:inputconditions}). Let $T' \in (\C^n)^{\otimes 3}$ such that $||T-T'|| \leq \delta $.
\par
Let $\Tilde{A}$ and $\Tilde{A'}$ be the output of Step 3 of Algorithm \ref{algo:robustSUBrecovery} on input $T$ and $T'$ respectively. Then 
\begin{equation}
    ||\Tilde{A} - \Tilde{A'}|| \leq 6\delta nk_FB^{\frac{3}{2}}. 
\end{equation}
\end{theorem}
\textbf{Step 4:} In the error-free scenario, the goal of this step is to take as input the ideal output $\Tilde{A}$ from Step 3 of the algorithm and find a semi-unitary matrix $S$ whose orthonormal columns span the range of $\Tilde{A}$. Step 4 of Algorithm \ref{algo:robustSUBrecovery} instead takes in as input the actual output $\Tilde{A'}$ from Step 3 and an accuracy parameter $\varepsilon$ and outputs $\Tilde{S} = \text{DEFLATE}(\Tilde{A'}, \varepsilon)$ where $\text{DEFLATE}$ is the algorithm mentioned in Theorem \ref{thm:deflate}. The goal of the remaining part of this section is to show using Theorem \ref{thm:deflate} that $||S - S'|| \leq \varepsilon$ with high probability.
\par
\underline{\textbf{Satisfying conditions of Theorem \ref{thm:deflate}:}}
We first need to show that $\Tilde{A}$ and $\Tilde{A'}$ indeed satisfy the conditions of Theorem \ref{thm:deflate}. To show this, we first show that $\Tilde{A}$ belongs to $X_r(n)$.
\par
Recall that if $(T,a)$ satisfy the $(r,n,B)$-input conditions (according to Definition \ref{def:inputconditions}), then $\Tilde{A} = 2k_FT^{(a)}(T^{(a)})^*$ where $T^{(a)} = \sum_{i=1}^n a_iT_i$ and $T_1,...,T_n$ are the slices of $T$. Now, since $(T,a)$ satisfy the $(r,n,B)$-input conditions, $\text{rank}(T^{(a)}) = r$. This implies that $\text{rank}(A) = \text{rank}(2k_FT^{(a)}(T^{(a)})^*) = r$. Since $A$ is a Hermitian matrix, we also deduce that $\text{rank}(A^2) = \text{rank}(A) =~r$. 
\par
Since $\kappa_F(T^{(a)}) \leq k_F$, using Property 1.(iii) of the pseudoinverse in Section~\ref{sec:pseudoinverse}, it follows that $$||\Big(T^{(a)}(T^{(a)})^*\Big)^{\dagger}|| \leq ||\Big((T^{(a)})^{\dagger}\Big)^*||||(T^{(a)})^{\dagger}|| \leq k_F.$$
Using part 1 of the Definition \ref{def:pseudoinverse}, we get that
\begin{align*}
 ||T^{(a)}(T^{(a)})^*|| &= ||T^{(a)}(T^{(a)})^*\Big(T^{(a)}(T^{(a)})^*\Big)^{\dagger}T^{(a)}(T^{(a)})^*||   \\
 &\leq ||T^{(a)}(T^{(a)})^*||^2||\Big(T^{(a)}(T^{(a)})^*\Big)^{\dagger}|| \leq k_F||T^{(a)}(T^{(a)})^*||^2.
\end{align*}
This finally gives us that 
\begin{equation}\label{eq:normA}
    ||{ \tilde{A}}|| = ||2k_F T^{(a)}(T^{(a)})^*|| \geq \frac{1}{2} > \frac{1}{3}.
\end{equation}
It remains to show that conditions (1) and (2) in the hypothesis of Theorem \ref{thm:deflate} are also satisfied.
\par
Set $\beta := 6\delta nk_FB^{\frac{3}{2}}$. From the hypothesis we know that $\delta \leq \frac{\varepsilon^4}{24 \times (20)^6n^7 B^{\frac{3}{2}}k_F}$  where $n \geq 1$ and $\varepsilon \leq 1$. Using this in Theorem \ref{thm:errSteps1-3}, we have:
\begin{equation}\label{eq:betabounds}
||\Tilde{A} - \Tilde{A'}|| \leq \beta = 6\delta nk_FB^{\frac{3}{2}} \leq \frac{\varepsilon^4}{4(20n)^6} \leq \frac{1}{4}.    
\end{equation}
Putting this back in (\ref{eq:normA}) and using triangle inequality, we get that $||{\tilde{A'}}|| \geq \frac{1}{4}$.
\par
\underline{\textbf{Applying Theorem \ref{thm:deflate}:}} Let $P$ be a semi-unitary matrix in $\C^{n \times r}$ such that its orthonormal columns span the range of $\Tilde{A}$. More formally, if $p_1,...,p_r$ are the columns of $P$, then $\text{span}\{p_1,...,p_r\} = \text{Im}(\Tilde{A})$. We firstly show that $P$ is indeed a solution to the SUB problem (refer to Definition \ref{def:ONBrecovery}) for the desired input tensor $T$. By hypothesis, we have that $(T,a)$ satisfies the $(r,n,B)$-input conditions. This gives us that $T$ is $r$-diagonalisable, that is, there exist linearly independent vectors $u_1,...,u_r \in \C^n$ such that $T = \sum_{i=1}^r u_i^{\otimes 3}$ and $\text{rank}(T^{(\alpha)}) = r$. Then by Lemma \ref{lem:redONB}, we already have that $\text{span}\{u_1,...,u_r\} = \text{Im}(T^{(\alpha)})$. Since by construction, we have that $\Tilde{A} = 2k_F T^{(\alpha)}(T^{(\alpha)})^*$, this also gives us that $\text{span}\{u_1,...,u_r\} = \text{Im}(\Tilde{A}) = \text{span}\{u_1,...,u_r\}$. 
\par
We apply Theorem~\ref{thm:deflate} on $A = \Tilde{A}$ and $A^{(0)} = \Tilde{A'}$ to get that Step 4 of Algorithm \ref{algo:robustSUBrecovery} outputs $\Tilde{P}$ such that $||P-\Tilde{P}|| \leq \varepsilon$. This shows that $\Tilde{P}$ is indeed a solution to the $\varepsilon$-approximate SUB recovery problem for the desired input tensor $T$.
\par
\underline{\textbf{Probability Analysis}:} We finally want to show that if $(T,a)$ is as mentioned in the hypothesis, Algorithm \ref{algo:robustSUBrecovery} outputs a solution to the $\varepsilon$-approximate SUB recovery problem for the desired input tensor $T$ with probability at least $1- \frac{1}{r}$.  
\par
\begin{lemma}\label{lem:singvalue}
Let $T$ be an order-$3$ $r$-diagonalisable symmetric tensor. Define $T^{(a)}$ to be a linear combination of the slices $T_1,...,T_n$ of $T$ such that $\kappa_F(T^{(a)}) < k_F$.
Then $\sigma_r \geq 2/r,$ where
$\sigma_1 \geq ... \geq \sigma_r$ are the $r$ non-zero eigenvalues of $ \Tilde{A} =~2k_FT^{(a)}(T^{(a)})^*$.
\end{lemma}
\begin{proof}
Since $\text{rank}(T^{(a)}) = r$, we know that there exist semi-unitary matrices $P,Q \in \C^{n \times r}$ such that $T^{(a)} = P \Sigma Q^*$ where $\Sigma$ is a diagonal matrix. Then $\Tilde{A} = 2 k_F (T^{(a)}(T^{(a)})^* = P (2k_F\Sigma^2) P^* $. Moreover, since $\kappa_F(T^{(a)}) \leq k_F$, this implies that $||(T^{(a)})^{\dagger}|| \leq \sqrt{k_F}$.  This gives us that $$||(\Tilde{A})^{\dagger}|| = ||\frac{1}{2k_F}((T^{(a)})^{\dagger})^*(T^{(a)})^{\dagger}|| \leq \frac{1}{2k_F}||(T^{(a)})^{\dagger}||^2 \leq \frac{1}{2} .$$ Let $2k_F\Sigma^2 = \text{diag}(\sigma_1,...,\sigma_r)$ where $\sigma_i \in \R_{+}$ are the eigenvalues of $\Tilde{A}$. 
Using this and the properties of the pseudoinverse of matrices, we get that
\begin{align*}
    \frac{1}{\sigma_r} \leq ||(2k_F\Sigma^2)^{\dagger}|| = ||P^* \Tilde{A}^{\dagger} P|| \leq ||P||^2||(\Tilde{A})^{\dagger}|| \leq \frac{r}{2}.
\end{align*}
The final inequality uses the fact that  $||P|| \leq~\sqrt{r}$ since $P$ is a semi-unitary matrix.
\end{proof}

Since $\beta \leq \frac{\varepsilon^4}{4 \times (20n)^6}$ by (\ref{eq:betabounds}),  Lemma \ref{lem:singvalue} implies that
\begin{equation}
   \frac{(20n)^3\sqrt{\beta}}{\varepsilon^2 \sigma_r(\Tilde{A})} \leq \frac{(20n)^3}{\varepsilon^2} \times \frac{2}{r} \times \frac{\varepsilon^2}{2 \times (20n)^3} \leq \frac{1}{r}. 
\end{equation}
This shows that if the desired input tensor satisfies the input conditions as mentioned in the hypothesis of Theorem \ref{thm:erroranalysis}, then Algorithm \ref{algo:robustSUBrecovery} indeed computes an $\varepsilon$-approximate solution to the SUB recovery problem with probability $1-\frac{1}{r}$.
{ Next we show that these input conditions are satisfied with high probability.}
\subsubsection{Probability Analysis { of the Input Conditions} }\label{sec:probability}

{ In Section~\ref{sec:probability} we show that 
the "input conditions" of Definition~\ref{def:inputconditions}
are satisfied with high probability.}
\begin{lemma}\label{lem:normnum}
Let $U = (u_{ij}) \in M_{n,r}(\mathbb{\C})$ be such that $\kappa_F(U) \leq B$. Then, given $\textbf{a} \in [-1,1]^n$, $\sum_{k \in [r]}|\langle a,u_k \rangle|^2 \leq nB$.
\end{lemma}
Recall from (\ref{eq:Frobeniuscondnumrect}), we had defined $\kappa_F(U) = ||U||_F^2 + ||U^{\dagger}||_F^2$. Also, recall that we had defined the discrete grid as $G_{\eta} = \{-1,-1+\eta,-1+2\eta,...,1-2\eta,1-\eta\}$ in (\ref{eq:discretegrid}) where $\frac{1}{\eta}$ is an integer for some $\eta \in (0,1)$.
\par
{We also need a result from \cite{KoiranSaha22} to gives a lower bound for a linear polynomial that occurs in the proof of Theorem \ref{thm:normprob}. Using anti-concentration inequalities due to \cite{CW01,forbes2017pspace,KoiranSaha22} and multivariate Markov's Theorem, the following theorem shows that the linear polynomial, evaluated on a point picked uniformly at random from the discrete grid previously described, is bounded far away from $0$ with high probability.
\begin{theorem}\label{thm:normdenom}[Theorem 7.13 in \cite{KoiranSaha22}]
Let $U \in \C^{r \times n}$ such that $\kappa_{F}(U) \leq B$ and let $u_1,...,u_r$ be the rows of $U$.
Then
\begin{align*}
    \text{Pr}_{\textbf{a} \in_U G_{\eta}^{n}}[|\langle a,u_i \rangle| \geq \frac{\alpha}{\sqrt{3B}} - \eta\sqrt{nB}] \geq 1- 2C_{CW}\alpha.
\end{align*}
for all $i \in [r]$, where $C_{CW}$ is the universal constant from the Carbery-Wright inequality \cite{CW01}.
\end{theorem}}
The following is the main theorem of this section.
\begin{theorem}\label{thm:normprob}
Let $T \in \C^n \otimes \C^n \otimes \C^n$ be an $r$-diagonalisable degree-$3$ symmetric tensor such that $\kappa(T) \leq B$, where $T_1,...,T_n$ are the slices of $T$. Let $a \in~[-1,1]^n$ be picked from $G^n_{\eta}$ uniformly at random and set $T^{(a)} := \sum_{i=1}^n a_iT_i$.
Then for all $k_F > nB^3$, we have that
\begin{align*}
    \text{Pr}_{\textbf{a} \in_U G^n_\eta}[\text{rank}(T^{(a)}) = r \text{ and } \kappa_F(T^{(a)}) \leq k_{F}] \geq  1- (2rC_{CW}\alpha_F + \frac{r\eta}{2})
\end{align*}
where $\alpha_{F} = \sqrt{3B}(\sqrt{\frac{rB^2}{k_F - nB^3}}+ \eta \sqrt{nB})$.
\end{theorem}
\begin{proof}
Let $U \in \C^{r \times n}$ be the matrix with rows $u_1,...,u_r$ such that $T = \sum_{i=1}^r u_i^{\otimes 3}$ and $\kappa(T) = \kappa_F(U) \leq B$. Since $|G_{\eta}| = \frac{2}{\eta}$, using Lemma \ref{lem:maxrankp3} for $S = G_{\eta}$
shows that when $a$ is picked uniformly and independently from grid $G^{n}_{\eta}$,  $\text{rank}(T^{(a)}) = r$ with probability at least $(1- \frac{r\eta}{2})$.
Using Corollary \ref{corr:P3structural}, and more precisely the fact that $T^{(a)} = U^TD^{(a)}U$, we have:
\begin{align*}
    ||(T^{(a)})^{\dagger}||_F &\leq ||U^{\dagger}||^2_F||(D^{(a)})^{-1}||_F \\
    &\leq  \kappa_F(U)||(D^{(a)})^{-1}||_F \\
    &\leq B||(D^{(a)})^{-1}||_F.
\end{align*}
Now,  $||(D^{(a)})^{-1}||_F^2 = \sum_{i=1}^r \frac{1}{|\langle a, u_i \rangle|^2}$. By Theorem \ref{thm:normdenom},  if $\textbf{a}$ is picked from $G^n_{\eta}$ uniformly at random, 
then for all $i \in [r]$, $|\langle a, u_i \rangle| \geq k$ with probability at least $1 - 2C_{CW}(\sqrt{3B}(k+ \eta \sqrt{nB}))$. This gives us that
\begin{align*}
    &\text{Pr}_{a \in G^n_{\eta}}[\exists m \in [r] |\langle a,u_m \rangle| \leq k \text{ } \cup  \text{ rank} (T^{(a)}) < r] \\
    &\leq \sum_{m=1}^r\text{Pr}_{a \in G^n_{\eta}}[|\langle a,u_m \rangle| \leq k] + \text{Pr}_{a \in G^n_{\eta}}[\text{rank} (T^{(a)}) < r] \\
    &\leq 2rC_{CW}(\sqrt{3B}(k+ \eta \sqrt{nB})) + \frac{r\eta}{2}.
\end{align*}
As a result,
\begin{align*}
    &\text{Pr}_{a \in G^n_{\eta}}[\text{for all } m \in [n] |\langle a,u_m \rangle| \geq k \text{ and }  \text{rank}(T^{(a)}) = r] \\
    &\geq 1- (2rC_{CW}(\sqrt{3B}(k+ \eta \sqrt{nB})) + \frac{r\eta}{2}).
\end{align*}
By Lemma \ref{lem:normnum},  $||D^{(a)}||^2 \leq nB$. This further implies that if $|\langle a,u_m \rangle| \geq k$ for all $m$, then  $||(D^{(a)})^{-1}||_F^2 + ||D^{(a)}||_F^2 \leq \frac{r}{k^2} + nB$, which in turn implies that $\kappa_F(T^{(a)}) = ||(T^{(a)})^{\dagger}||_F^2 + ||T^{(a)}||^2_F \leq \frac{rB^2}{k^2} + nB^3 $. Setting $k = \sqrt{\frac{rB^2}{k_F - nB^3}}$ gives the desired conclusion.
\end{proof}
\subsubsection{Finishing the proof of Theorem \ref{thm:SUBrecoverymaintheorem}}\label{sec:finishproofrobustsubrecovery}
\underline{\textbf{Setting bounds for $k_F$ and $\eta$ from Algorithm \ref{algo:robustSUBrecovery} and Theorem \ref{thm:normprob}:}}
We set
\begin{equation}\label{eq:kappafandeta}
    k_F = (192C^2_{CW} + 1)n^5B^3 \text{ and } \eta = \frac{1}{2C_{CW}r^2\sqrt{nB}}.
\end{equation}
 Since $nB^3 <~r^5B^3$, we have
\begin{align*}
    \alpha_F &= \sqrt{3B}(\sqrt{\frac{rB^2}{k_F - nB^3}} + \eta \sqrt{nB}) \\
    &=\sqrt{3B}(\sqrt{\frac{1}{192C^2_{CW}r^4B}} +  \frac{1}{2C_{CW}r^2\sqrt{nB}} \sqrt{nB}) \\
    &\leq \frac{1}{8C_{CW}r^2} +\frac{1}{2C_{CW}r^2.} \\
    &\leq \frac{5}{8C_{CW}r^2}
\end{align*}
This gives us that 
\begin{equation}\label{eq:probabilitysetting}
\begin{split}
    2rC_{CW}\alpha_F &\leq \frac{5}{4r} \\
    \frac{\eta r}{2} &= \frac{1}{4C^2_{CW}r\sqrt{nB}} \leq \frac{1}{4C^2_{CW}r^{\frac{3}{2}}}
\end{split}
\end{equation}
\underline{\textbf{Setting bounds for }$\delta$:}
Since, $\delta \leq \frac{\varepsilon^4}{12 \times (20)^6 \times n^7B^{\frac{3}{2}}k_F}$, using the value of $k_F$ from (\ref{eq:kappafandeta}), we get that
\begin{equation}\label{eq:delta}
    \delta \leq \frac{\varepsilon^4}{c_{\delta} \times n^{12}B^{\frac{9}{2}} }
\end{equation}
where $c_{\delta} = 24 \times (20)^6 \times (192C^2_{CW} + 1)$.
\par
Let $E_1$ be the event that Algorithm \ref{algo:robustSUBrecovery} returns a solution to the $\varepsilon$-approximate SUB recovery problem for the given $r$-diagonalisable tensor $T$. Let $T^{(a)} = \sum_{i=1}^n a_iT_i$ be as computed in Step 1 of the algorithm where $T_1,...,T_n$ are the slices of $T$. We define $E_2$ to be the event  that $\text{rank}(T^{(a)}) = r$ and $E_3$ to be the event that $\kappa_F(T^{(a)}) \leq k_F$ as defined in (\ref{eq:kappafandeta}).
\par
{ 
Suppose that the algorithm is given as input  a tensor $T' \in~\C^n \otimes~\C^n \otimes~\C^n$ such that $||T-T'|| \leq \delta$, where $\delta$ is within the bounds  set in (\ref{eq:delta}).
Suppose moreover that $ \text{rank}(T^{(a)}) = r$ and  $\kappa_F(T^{(a)}) \leq k_F$.
By Theorem \ref{thm:erroranalysis}, the algorithm returns a solution to the $\varepsilon$-approximate SUB recovery problem for the given tensor $T$ with probability $1 - \frac{1}{r}$, where the randomness comes from the internal random choices of DEFLATE algorithm mentioned in Theorem~\ref{thm:deflate}. }

This gives us that
\begin{equation}\label{eq:probE1}
    \text{Pr}[E_1| E_2 \cap E_3] \geq 1 - \frac{1}{r}.
\end{equation}
Using Theorem \ref{thm:normprob} and the bounds from (\ref{eq:probabilitysetting}), we get that
\begin{equation}\label{eq:probe2e3}
    \text{Pr}_{a \in G^n_{\eta}} [E_2 \cap E_3] \geq 1 - (\frac{5}{4r} + \frac{1}{4C^2_{CW}r^{\frac{3}{2}}}).
\end{equation}
Combining (\ref{eq:probE1}) and (\ref{eq:probe2e3}), we get that 
\begin{align*}
    \text{Pr}[E_1] \geq (1 - \frac{1}{r})(1 - (\frac{5}{4r} + \frac{1}{4C^2_{CW}r^{\frac{3}{2}}}))
\end{align*}
\newline
\underline{\textbf{Complexity Analysis: }}
We analyse the steps of the algorithm and deduce the number of arithmetic operations required to perform the algorithm.
\begin{enumerate}
    \item In Step 1, we need to compute a linear combination of the slices of the tensor $T$. This can be done by performing inner products of the form $(T^{(a)})_{ij} = \sum_{k=1}^n a_k(T_k)_{ij}$ for all $i,j \in [n]$, each of which requires $O(n)$ arithmetic operations. Hence, this entire step can be performed using $O(n^3)$ arithmetic operations.
    \item In Step 2, $A = T^{(a)}(T^{(a)})^*$ can be computed using $O(n^3)$ arithmetic operations.
    \item In Step 3, $\Tilde{A} = 2k_F A$ can be computed using $O(n^2)$ arithmetic operations.
    \item By Theorem \ref{thm:deflate},  Step 4 can be executed in $O(T_{MM}(n))$ arithmetic operations.
\end{enumerate}
\section{Undercomplete Tensor Decomposition} \label{sec:algoundercomplete}

\subsection{Algebraic Algorithm}\label{sec:algalgoundercomplete}
{ In this subsection, we reduce the undercomplete tensor decomposition problem to the complete case assuming that the input $r$-diagonalisable tensor $T$ is given exactly and Algorithm \ref{algo:ONBrecovery} computes an exact solution to the SUB recovery problem for $T$. We show that if we further assume that Algorithm \ref{algo:completeexact} returns an exact solution to the tensor decomposition problem for a diagonalisable tensor with high probability, then indeed the following algorithm returns an exact solution to the tensor decomposition problem for the $r$-diagonalisable tensor with high probability as well.}

\begin{algorithm}[H] \label{algo:undercompleteexactalgebraic}
\SetAlgoLined
\nonl \textbf{Input:} An order-3 $r$-diagonalisable symmetric tensor $T \in \C^{n \times n \times n}$.  \\
\nonl \textbf{Output:} vectors $l_1,...,l_r$ such that $T = \sum_{i=1}^r l_i^{\otimes 3}$ \\
Let $P { \in \mathbb{C}^{n \times r}}$ be the output of Algorithm \ref{algo:ONBrecovery} on $T$. \\
\nonl Let $T' = (\overline{P} \otimes \overline{P} \otimes \overline{P}).T$. \\
Let $u_1,...,u_r$ be the output of Algorithm \ref{algo:completeexact} on $T'$  for $n = r$. \\
Compute  $l_i = Pu_i$ for all $i \in [r]$. \\
\nonl Output $l_1 ,...,l_r$.
\caption{ Undercomplete decomposition of  symmetric tensors in the exact setting.}
\end{algorithm}
{ In  Section~\ref{sec:undercompleterobust} we will refine
this idealized algorithm into a robust approximate algorithm.}
\begin{theorem}\label{thm:erroranalysisexactarithmetic}
If the input tensor $T \in \C^n \otimes \C^n \otimes \C^n$ is $r$-diagonalisable, then Algorithm \ref{algo:undercompleteexactalgebraic} returns a decomposition with high probability. More formally, 
the algorithm returns linearly independent vectors $u_1,...,u_r \in \C^n$ such that $T = \sum_{i=1}^r u_i^{\otimes 3}$ with probability $(1- \frac{1}{2r})^2$.
\end{theorem}
The correctness proof of this theorem { proceeds roughly as follows.} 
If $P$ is a solution to the SUB recovery problem for an $r$-diagonalisable tensor $T \in \C^n \otimes \C^n \otimes \C^n$, then $T' = (\overline{P}\otimes \overline{P} \otimes \overline{P}) \in \C^r \otimes \C^r \otimes \C^r$ is a diagonalisable tensor. Moreover, a unique decomposition of $T'$ can be used to compute a decomposition of $T$ as well.
{ This is made precise in Lemma~\ref{lem:equivalencelem} below.
Before stating this lemma we} recall that for if $T = \sum_{i=1}^r u_i^{\otimes 3}$ for some $u_i \in \C^n$, then for any matrix $M \in \C^{n \times m}$, the change of basis operation $(M \otimes M \otimes M).T$ is defined as $(M \otimes M \otimes M).T = \sum_{i=1}^r (M^Tu_i)^{\otimes 3}$.
\par
In the next lemma, we show how a solution to the SUB recovery problem for an $r$-diagonalisable tensor in $\C^n \otimes \C^n \otimes \C^n$ can be used to transform it to a diagonalisable tensor in $\C^r \otimes \C^r \otimes \C^r$.
\begin{lemma}\label{lem:equivalencelem}
Let $T \in (\mathbb{C}^n)^{\otimes 3}$ be an $r$-diagonalisable tensor. Then there exist linearly independent vectors $u_i \in \C^n$ such that $T = \sum_{i=1}^r u_i^{\otimes 3}$. 
Let $P$ be a solution to the SUB recovery problem for $T$ (according to Definition~\ref{def:ONBrecovery}). Define $T' := (\overline{P} \otimes \overline{P} \otimes \overline{P}).T$.
Then the following holds:
\begin{itemize}
    \item $T'$ is diagonalisable.  More formally, $T' = \sum_{i=1}^r a_i^{\otimes 3}$ where the vectors $a_i = P^*u_i \in \C^r$ are linearly independent.
    \item If $T' = \sum_{i=1}^r v_i^{\otimes 3}$ where the $v_i$ are linearly independent, then $T =~\sum_{i=1}^{ r}(Pv_i)^{\otimes 3}$.
\end{itemize}
\end{lemma}
\begin{proof}
From the definition of the change of basis operation, it follows that $T' = (\overline{P} \otimes \overline{P} \otimes \overline{P}).T = \sum_{i=1}^r (a_i)^{\otimes 3}$ where $a_i = P^*u_i \in \C^r$. It remains to show that the  $a_i$ are linearly independent.
\par
Let $p_1,...,p_r$ be the columns of $P$. By Definition \ref{def:ONBrecovery},
$P \in \C^{n \times r}$ is a semi-unitary matrix and $\text{span}(p_1,...,p_r) = \text{span}(u_1,...,u_r).$ Let $B \in \text{GL}_r(\C)$ be the associated change of basis matrix and let $U$ be the matrix with columns $u_1,...u_r$. Writing this in terms of matrices, we have
\begin{equation}\label{eq:UPB}
    U = PB.
\end{equation}
Now let $A \in \C^{r \times r}$ be the matrix with columns $a_i = P^*u_i$. Then $A = P^* U$. Recall that $P^*P = I_r$ since
$P$ is semi-unitary.
Hence $A = P^* PB = B$ by~(\ref{eq:UPB}). This implies that $A$ is an invertible matrix and its columns $a_1,\ldots,a_r$ are therefore linearly independent.
\par
For the second part, let $V$ be the matrix with columns $v_1,...,v_r \in \C^r$. Since $T$ is $r$-diagonalisable and $T' = (\overline{P} \otimes \overline{P} \otimes \overline{P}).T$ is diagonalisable, then using the uniqueness of decomposition for $r$-diagonalisable tensors, we get that there exist vectors $u'_1,...,u'_r \in \C^n$ such that $T = \sum_{i=1}^r (u'_i)^{\otimes 3}$ and $P^*u'_i = v_i$ for all $i \in [r]$. Let $U'$ be the matrix with columns $u'_1,...,u'_r$. 
In terms of matrices, this gives us that 
\begin{equation}\label{eq:p*u'}
    P^*U' = V.
\end{equation}
Moreover since $P$ is a solution to the SUB recovery problem for $T$, using a similar argument as in the previous part, there must exist an invertible matrix $B' \in \text{GL}_r(\C)$ such that
\begin{align*}
    U' = PB'.
\end{align*}
Combining this with (\ref{eq:p*u'}), we have $V = P^*PB' = B'$ which in turn gives us
the desired result: $u'_i = Pv_i$. 
\end{proof}

\begin{proof}[Proof of Theorem \ref{thm:erroranalysisexactarithmetic}]
Let $E_1$ denote the event that Algorithm \ref{algo:undercompleteexactalgebraic} returns linearly independent vectors $u_1,...,u_r \in \C^n$ such that $T = \sum_{i=1}^r u_i^{\otimes 3}$. We want to show that $\text{Pr}[E_1]$ is large. Let $E_2$ be the event that in Step 1, Algorithm \ref{algo:ONBrecovery} returns a semi-unitary matrix $P \in \C^{n \times r}$ such that $P$ is a solution to the SUB recovery problem for the input tensor $T$. Let $E'_2$ be the event such that the tensor $T' = (P\otimes P \otimes P).T$ is diagonalisable.
\par
From Lemma \ref{lem:equivalencelem}, it follows that $\text{Pr}[E_2] \leq \text{Pr}[E'_2]$.
If the set of size $S$ from which the internal random bits of Algorithm \ref{algo:ONBrecovery} are picked has size $|S| = 2r^2$, using Theorem \ref{thm:ONBrecovery}, we  have 
\begin{equation}\label{eq:E2}
    \text{Pr}[E_2] \geq 1-\frac{1}{2r}.
\end{equation}
Moreover, using Theorem \ref{thm:Jennrichexact} for $T'$ and $n = r$, we get that 
\begin{equation}\label{eq:undercompletee3}
    \text{Pr}[E_1|E'_2] \geq 1-\frac{1}{2r}.
\end{equation}
Combining this with (\ref{eq:E2}), we get that
\begin{equation}
    \text{Pr}[E_1] \geq \text{Pr}[E_1|E'_2]\text{Pr}[E'_2] \geq (1-\frac{1}{2r}) \text{Pr}[E_2] \geq (1- \frac{1}{2r})^2.
\end{equation}
\end{proof}

\subsection{Robust Algorithm for Undercomplete Tensor Decomposition}\label{sec:undercompleterobust}

\begin{definition}[$\delta$-forward approximation problem for undercomplete tensor decomposition]\label{def:undercompleteproblemdef}
Let $T$ be an input tensor that can be decomposed as $\sum_{i=1}^r u_i^{\otimes 3}$ where the $u_i$'s are linearly independent. The goal is to find
linearly independent vectors $u'_1,...,u'_r$, 
such that there exists a permutation $\pi \in S_n$ satisfying
\begin{align*}
    ||\omega_iu_{\pi(i)} - u'_i|| \leq \delta
\end{align*}
where $\omega_i$ is a cube root of unity. 
\end{definition}
The next theorem justifies the importance of finding a solution to the semi-unitary basis recovery problem. That is, if $P$ is a solution to the SUB recovery problem for the input $r$-diagonalisable tensor $T$, then $T' = (\overline{P} \otimes \overline{P} \otimes \overline{P}).T$ is indeed diagonalisable and the condition number of $T'$ is the same as the condition number for $T$.
\begin{theorem}\label{thm:reductionconditionnumber}
Let $T \in (\mathbb{C}^n)^{\otimes 3}$ be an $r$-diagonalisable symmetric tensor. Let $P \in \C^{n \times r}$ be a unitary matrix which is a solution to the SUB recovery problem for $T$ as in Definition \ref{def:ONBrecovery}. 
Then $\kappa(T) = \kappa(T')$, where $T' \in (\mathbb{C}^r)^{\otimes 3}$ is the tensor $(\overline{P} \otimes \overline{P} \otimes \overline{P}).T$. 
\end{theorem}
\begin{proof}
For any matrix $M \in \text{GL}_r(C)$, we define $\Tilde{M} = \begin{pmatrix}
M \\ 0 \end{pmatrix} \in \C^{n \times r}$. Using Property 3 of the Moore-Penrose inverse (according to Section \ref{sec:pseudoinverse}), it can be observed that $(\Tilde{M})^{\dagger} = \begin{pmatrix}
M^{-1} & 0 \end{pmatrix} \in \C^{r\times n}$. Let $U \in \C^{r \times n}$ be the matrix with rows $u_1,...,u_r$. Then $\kappa(T) = ||U||^2_F + ||U^{\dagger}||^2_F$. Let $p_1,...,p_r \in \C^n$ be the columns of $P$. Define matrix $\Tilde{P} \in \C^{n \times n}$ where its first $r$ columns are the columns of $P$ and the columns of $\Tilde{P}$ complete an orthonormal basis for~$\C^n$. 
From Lemma \ref{lem:equivalencelem}, the vectors $a_i = P^*u_i \in \C^r$  are linearly independent. Since $\text{span}\{u_1,...,u_r\} = \text{span}\{p_1,...,p_r\}$,
$\Tilde{P}^*u_i = \Tilde{a_i}$ where $\Tilde{a_i} = (a_i,0,...,0) \in \C^n$.
Writing this in matrix notation, we have $\Tilde{P}^*U^T = (A^{(0)})^T$ where $A^{(0)}$ is the matrix with rows $\Tilde{a_1},...,\Tilde{a_n}$. Then by Lemma \ref{lem:equivalencelem},
\begin{equation}\label{eq:kappaTT'}
\begin{split}
    \kappa(T') = \kappa_F(A) &= ||A||_F^2 + ||A^{-1}||_F^2 \\
    &= ||(A^{(0)})^T||_F^2 + ||((A^{(0)})^T)^{\dagger}||_F^2 = ||\Tilde{P}^*U^T||_F^2 + ||(\Tilde{P}^*U^T)^{\dagger}||_F^2.
\end{split}
\end{equation}
Multiplying a matrix by a unitary matrix keeps the Frobenius norm unchanged. Since $\Tilde{P}^*$ is unitary, $||\Tilde{P}^*U^T||_F = ||U^T||_F = ||U||_F$. Also, since $\Tilde{P}^*$ has orthornomal columns, by Property 1(a) of the Moore-Penrose inverse in Section \ref{sec:pseudoinverse}, we have $(P^*U^T)^{\dagger} = (U^T)^{\dagger}(P^*)^{\dagger} = (U^{\dagger})^T P$. Using again the fact that $P$ is a unitary matrix, $||(U^{\dagger})^T \Tilde{P}||_F = ||(U^{\dagger})^T||_F = ||U^{\dagger}||_F$. This gives us that  $||(\Tilde{P}^*U^T)^{\dagger}||_F = ||U^{\dagger}||_F$. Putting this back in (\ref{eq:kappaTT'}), we obtain
\begin{align*}
    \kappa(T') = ||\Tilde{P}^*U^T||_F^2 + ||(\Tilde{P}^*U^T)^{\dagger}||_F^2 = ||U||^2_F + ||U^{\dagger}||^2_F = \kappa(T).
\end{align*}
\end{proof}

\begin{algorithm}[H] \label{algo:undercompleteexact}
\SetAlgoLined
\nonl \textbf{Input:} An order-3 $r$-diagonalisable symmetric tensor $T \in \C^n \otimes \C^n \otimes \C^n$, estimate $B \geq \kappa(T)$ and accuracy parameter $\varepsilon$.  \\
\nonl \textbf{Output:} vectors $l_1,...,l_r \in \C^n$ such that $T = \sum_{i=1}^r l_i^{\otimes 3}$ \\
\nonl Let $\varepsilon_1 = \frac{1}{r^{C_1\log^4(\frac{rB}{\varepsilon})}}$ \\
Let $P { \in \C^{n \times r}}$ be the output of Algorithm \ref{algo:robustSUBrecovery} on input $(T,B,\varepsilon_1)$. \\
Compute $S = (\overline{P} \otimes \overline{P} \otimes \overline{P}).T { \in \C^{r \times r \times r}}$ using Algorithm \ref{algo:changeofvariablesexact}. \\

Let $u_1,...,u_r { \in \C^r}$ be the output of Algorithm \ref{algo:Jennrich} on input $S$, 
with estimate $B$ for $\kappa(S)$ and accuracy parameter $\varepsilon_3 = \frac{\varepsilon}{2\sqrt{r}}$. \\
Compute  $l_i = Pu_i {\in \C^n}$ for all $i \in [r]$. \\
\nonl Output $l_1 ,...,l_r$.
\caption{Approximate algorithm for undercomplete decomposition of  symmetric tensors.}
\end{algorithm}
The following is the main theorem of this section. {It shows that if the input $r$-diagonalisable tensor $T$ has noise which is inverse \textit{quasi-polynomially} bounded (in $n,B$ and $\frac{1}{\varepsilon}$), where $\varepsilon$ is the desired accuracy parameter), then the algorithm returns an $\varepsilon$-approximate solution to the tensor decomposition problem for $T$ with high probability.}
\begin{theorem}\label{thm:undercompleteexactproof}
Let $T \in \C^n \otimes \C^n \otimes \C^n$ be an $r$-diagonalisable tensor for some $r \leq n$ and let $T' \in \C^n \otimes \C^n \otimes \C^n$ 
 be such that $$||T - T'|| \leq \delta \in~(0, \frac{1}{\text{poly}(n,B)r^{4C_1\log^4(\frac{rB}{\varepsilon})}})$$ for some constant $C_1$.\footnote{The exact values for $\text{poly}(n,B)$ and $C_1$ are set in (\ref{eq:delta}) and (\ref{eq:eps}) respectively.} Then on input $T'$, a desired accuracy parameter $\varepsilon$ and some estimate $B \geq \kappa(T)$, Algorithm \ref{algo:undercompleteexact} outputs an $\varepsilon$-approximate solution to the tensor decomposition problem for $T$ with probability at least
$$\Big(1 - \frac{13}{r^2}\Big)^2\Big(1 - (\frac{5}{4r} + \frac{1}{4C^2_{CW}r^{\frac{3}{2}}})\Big).$$
The algorithm requires $O(n^{\omega(2)} + T_{MM}(r)\log^2(\frac{rB}{\varepsilon}))$ arithmetic operations.
\end{theorem}
Note here that $T_{MM}(r) = r^{\omega + \eta} < r^3$ and $\omega(2) = 3.251640$.
\begin{lemma}\label{lem:vectortensorbounds}
Let $a,b \in \C^n$ be such that $||a|| \leq l$ and $||a-b|| \leq \eta < l$. Then $||a^{\otimes 3} - b^{\otimes 3}||_F < 7\eta l^2 n^{\frac{3}{2}}$.
\end{lemma}
\begin{proof}
For all $i \in [n]$, $|a_i - b_i| \leq ||a - b|| \leq \eta$. Since $|a_i| \leq ||a|| \leq k$, we get that $|b_i| \leq \eta + l$. From this and from the definition of the Frobenius norm for tensors,
\begin{align*}
    ||a^{\otimes 3} - b^{\otimes 3}||_F^2 &= \sum_{i,j,k =1}^n |a_ia_ja_k - b_ib_jb_k|^2 \\
    &\leq \sum_{i,j,k =1}^n \Big(|a_ia_ja_k - b_ia_ja_k| + |b_ia_ja_k - b_ib_ja_k| + |b_ib_ja_k - b_ib_jb_k|\Big)^2 \\
    &\leq \sum_{i,j,k =1}^n \eta^2\Big( |a_ja_k| + |b_ia_k| + |b_ib_j| \Big)^2 \\
    &\leq \sum_{i,j,k =1}^n \eta^2(l^2 + l(\eta + l) + (\eta + l)^2)^2 \\
    &< \sum_{i,j,k =1}^n \eta^2(7l^2)^2 = \eta^2 n^3 (7l^2)^2.
\end{align*}
This finally gives us that
\begin{align*}
    ||a^{\otimes 3} - b^{\otimes 3}||_F \leq 7\eta l^2 n^{\frac{3}{2}}. 
\end{align*}
\end{proof}

\begin{lemma}\label{lem:errstep2tensordecomposition}
Let $T \in \C^n \otimes \C^n \otimes \C^n$ be an order-$3$ $r$-diagonalisable tensor where $\kappa(T) \leq B$ and let $T' \in \C^n \otimes \C^n \otimes \C^n$ be such that $||T - T'||_F \leq \delta$. Let $P \in \C^{n \times r}$ be a semi-unitary matrix which is a solution to the SUB recovery problem for $T$ and let $P'$ be an $n \times r$ matrix such that $||P-P'|| \leq \varepsilon_1 < 1$. Let $S = (\overline{P} \otimes \overline{P} \otimes \overline{P}).T$ and $S' = (\overline{P'} \otimes \overline{P'} \otimes \overline{P'}).T'$. Then 
\begin{equation}
    ||S - S'||_F \leq 7\varepsilon_1  r^{\frac{7}{2}}B^{\frac{3}{2}} + 8\delta r^{\frac{3}{2}}.
\end{equation}
\end{lemma}
\begin{proof}
Since $T$ is an $r$-diagonalisable tensor, it has a decomposition of the form $T = \sum_{i=1}^r u_i^{\otimes 3}$ where $u_1,...,u_r \in \C^n$ are linearly independent. Let $U$ be the matrix with rows $u_1,...,u_r$. It follows from the hypothesis that $\kappa_F(U) = \kappa(T) \leq B$. Using the definition of operator norm on matrices and the fact that $||A|| = ||A^*||$ for any matrix, we have that for all $i \in [r]$, $||P^*u_i - P'^*u_i|| \leq ||P - P'||||u_i|| \leq \varepsilon_1||U||_F \leq \varepsilon_1\sqrt{B}$. Moreover, we also have that $||P^*u_i|| \leq ||P||||u_i|| \leq \sqrt{rB}$.
Using this, along with the definition of the change of basis operation and Lemma \ref{lem:vectortensorbounds} for $n = r$, we have 
\begin{equation}\label{eq:step52}
\begin{split}
        ||(\overline{P} \otimes \overline{P} \otimes \overline{P}).T -  (\overline{P'} \otimes \overline{P'} \otimes \overline{P'}).T|| &\leq \sum_{i=1}^r ||(P^*u_i)^{\otimes 3} - (P'^*u_i)^{\otimes 3} || \\
        &\leq \sum_{i=1}^r  7\varepsilon_1\sqrt{B} \cdot rB \cdot r^{\frac{3}{2}} = 7 \varepsilon_1 r^{\frac{7}{2}}B^{\frac{3}{2}}.
\end{split}
\end{equation}
Let $p'_1,...,p'_r$ be the columns of $P'$. Since $||P||_F \leq \sqrt{r}$ and $||P - P'|| \leq \varepsilon_1$, we have $||P'||_F \leq \varepsilon_1 + \sqrt{r} \leq 2\sqrt{r}$. Then
\begin{equation}\label{eq:Step51}
\begin{split}
&||(\overline{P'} \otimes \overline{P'} \otimes \overline{P'}).T - (\overline{P'} \otimes \overline{P'} \otimes \overline{P'}).T'||^2 \\
&= \sum_{i_1,i_2,i_3 = 1}^r  \Big|\sum_{j_1,j_2,j_3 = 1}^n \overline{P'_{j_1i_1}}\overline{P'_{j_2i_2}}\overline{P'_{j_3i_3}} (T'_{j_1j_2j_3} - T_{j_1j_2j_3})\Big|^2 \\
&\leq \sum_{i_1,i_2,i_3 = 1}^r ||T-T'||_F^2 \Big(\sum_{j_1,j_2,j_3 = 1}^n |  \overline{P'_{j_1i_1}}\overline{P'_{j_2i_2}}\overline{P'_{j_3i_3}}|^2\Big) \\
&= \sum_{i_1,i_2,i_3 = 1}^r ||T-T'||_F^2||\overline{p'_{i_1}} \otimes \overline{p'_{i_2}} \otimes \overline{p'_{i_3}}||_F^2 \\
&\leq \delta^2 \sum_{i_1,i_2,i_3 = 1}^r ||\overline{p'_{i_1}} \otimes \overline{p'_{i_2}} \otimes \overline{p'_{i_3}}||_F^2 \\
&= \delta^2 \sum_{i_1,i_2,i_3 = 1}^r ||\overline{p'_{i_1}}||^2||\overline{p'_{i_2}}||^2 ||\overline{p'_{i_3}}||^2 \leq \delta^2 \Big(\sum_{i=1}^r ||\overline{p'_i}||^2\Big)^3 \leq 64\delta^2 r^3. 
\end{split}
\end{equation}
The first equality follows from (\ref{eq:changeofbasisdef}). Combining (\ref{eq:Step51}) and (\ref{eq:step52}) gives us the desired result.
\end{proof}
\begin{proof}[Proof of Theorem \ref{thm:undercompleteexactproof}]
We first show that indeed Algorithm \ref{algo:undercompleteexact} returns an $\varepsilon$-approximate solution to the tensor decomposition problem. As mentioned in the hypothesis, $T \in (\C^n)^{\otimes 3}$ is an $r$-diagonalisable tensor, which we will refer to as the desired input of the algorithm and $T' \in (\C^n)^{\otimes 3}$ is another tensor "close" to $T$, which is the actual input to the algorithm. 
\newline
\underline{\textbf{Error Analysis}:}
\par
Let $B$ be the given upper bound for the condition number of the desired input tensor $T$ and $\varepsilon$ be the desired accuracy parameter given to the algorithm. We define the error bounds at the end of Step $i$ of Algorithm \ref{algo:undercompleteexact} as follows:
\begin{equation}\label{eq:eps}
    \varepsilon_i := \begin{cases}
        \varepsilon & \text{if } i = 4\\
        \frac{\varepsilon}{2\sqrt{r}} & \text{if } i = 3 \\
        \frac{1}{r^{C_i \log^4 \frac{rB}{\varepsilon}}} & \text{ if } i \in \{1,2\}
    \end{cases}
\end{equation}
where $C_1 = 2C_2$ and $C_2$ is a constant we fix later in (\ref{eq:C2}). 
\newline
\underline{\textbf{Step 1:}} By Theorem \ref{thm:SUBrecoverymaintheorem}, if $$||T - T'|| \leq \delta < \frac{\varepsilon_1^4}{\text{poly}(n,B)} = \frac{1}{\text{poly}(n,B)r^{4C_1 \log^4 \frac{rB}{\varepsilon}}},$$ then Step 1 of the Algorithm returns a matrix $P'$ which is a solution to the $\varepsilon_1$-robust recovery problem for $T$ (the probability of success will be computed later).  More formally, let $P$ be the actual solution to the SUB recovery problem for the desired input $T$. Then Step 1 of the Algorithm returns a matrix $P'$ such that $||P - P'|| \leq \varepsilon_1$.
\newline
\underline{\textbf{Step 2:}} Let $S = (\overline{P} \otimes \overline{P} \otimes \overline{P}).T$ and $S' = (\overline{P'} \otimes \overline{P'} \otimes \overline{P'}).T'$. By Lemma \ref{lem:errstep2tensordecomposition},
\begin{equation}
    ||S - S'|| \leq 7\varepsilon_1r^{\frac{7}{2}}B^{\frac{3}{2}} + 8\delta r^{\frac{3}{2}} \leq 15\varepsilon_1r^{\frac{7}{2}}B^{\frac{3}{2}} \leq \frac{1}{r^{\frac{C_1}{2} \log^4 \frac{rB}{\varepsilon}}} = \varepsilon_2
\end{equation}
where the last inequality used the fact that $\delta \leq \varepsilon_1 \leq 1$ and $r \geq 1$.
\newline
\underline{\textbf{Step 3:}} Since $P$ is a solution to the SUB recovery problem for $T$, 
if $\kappa(T) \leq B$ then $\kappa(S) = \kappa(T) \leq B$ by Theorem \ref{thm:reductionconditionnumber}.
Next, we apply Theorem \ref{thm:completetensordecomposition} on input $S \in \C^r \otimes \C^r \otimes \C^r$, estimate $B$ for the condition number of $S$ and accuracy parameter $\varepsilon_3 = \frac{\varepsilon}{2\sqrt{r}}$.  The theorem implies  that if Algorithm \ref{algo:Jennrich} gets as input some tensor $S' \in \C^r \otimes \C^r \otimes \C^r$ such that $||S-S'|| \leq \frac{1}{r^{C \log^4 \frac{rB}{\varepsilon}}}$ for some constant $C$, then it outputs an $\varepsilon$-approximate solution to the tensor decomposition problem for $S$ with probability $(1 - \frac{1}{r} - \frac{12}{r^2})(1 - \frac{1}{\sqrt{2r}} - \frac{1}{r})$. 
\par
We finally set 
\begin{equation}\label{eq:C2}
    C_2 = 40C.
\end{equation}
This gives us that $\varepsilon_2 = \frac{1}{r^{40C \log^4 \frac{rB}{\varepsilon}}} \leq \frac{1}{r^{C \log^4 \frac{\sqrt{2}r^{\frac{3}{2}}B}{\varepsilon}}} = \frac{1}{r^{C \log^4 \frac{rB}{\varepsilon_3}}}.$
Let $u_1,...,u_r \in~\C^r$ be linearly independent vectors such that $S = \sum_{i=1}^r u_i^{\otimes 3}$. By Theorem \ref{thm:completetensordecomposition},  Algorithm \ref{algo:Jennrich} returns vectors $u'_i \in \C^r$ such that
\begin{equation}\label{eq:uiu'ieps3}
    ||u_i - u'_i|| \leq \varepsilon_3.
\end{equation}
\underline{\textbf{Step 4:}} Let $l_i = Pu_i$ be the actual output of Step 4 of the algorithm and let $l'_i = Pu'_i$. Since the matrix $U$ with columns $u_1,...,u_r$ diagonalises the tensor $S$, using the fact that $||u_i|| \leq ||U||_F \leq \sqrt{\kappa(S)} \leq \sqrt{B}$ along with (\ref{eq:uiu'ieps3}), we have:
\begin{align*}
    ||l_i - l'_i|| \leq ||Pu_i - Pu'_i|| + ||Pu'_i - P'u'_i|| &\leq ||P||||u_i - u'_i|| + ||P-P'||||u'_i|| \\
    &\leq \sqrt{r}\varepsilon_3 + \varepsilon_1(\sqrt{B} + \varepsilon_3) \leq 2\sqrt{r}\varepsilon_3 = \varepsilon.
\end{align*}
The last equality follows from the definition of $\varepsilon_3$ in (\ref{eq:eps}) and the fact that $\varepsilon_1\sqrt{B} \leq \sqrt{r}\varepsilon_3$.
\newline
\textbf{\underline{Probability Analysis:}} Let $E_1$ be the event that Step 1 of Algorithm \ref{algo:undercompleteexact} computes an $n \times r$ matrix $P$ which is an $\varepsilon_1$-robust solution to the SUB recovery problem. 
By Theorem \ref{thm:SUBrecoverymaintheorem}, 
\begin{equation}\label{eq:probe1}
    \text{Pr}[E_1] \geq (1 - \frac{1}{r})(1 - (\frac{5}{4r} + \frac{1}{4C^2_{CW}r^{\frac{3}{2}}})).
\end{equation}
Let $E_2$ be the event that Algorithm \ref{algo:undercompleteexact} indeed returns an $\varepsilon$-approximate solution to the tensor decomposition problem for the desired input tensor~$T$. By the above error analysis,  if Step 1 of the algorithm returns an $\varepsilon_1$-robust solution to the SUB recovery problem, then Algorithm \ref{algo:undercompleteexact} returns a solution to the $\varepsilon$-approximate tensor decomposition problem for the desired input tensor $T$ with probability at least 
$(1 - \frac{1}{r} - \frac{12}{r^2})(1 - \frac{1}{\sqrt{2r}} - \frac{1}{r})$. This can be written as
\begin{align*}
    \text{Pr}[E_2 | E_1] \geq (1 - \frac{1}{r} - \frac{12}{r^2})(1 - \frac{1}{\sqrt{2r}} - \frac{1}{r}).
\end{align*}
Combining this with (\ref{eq:probe1}), we get that
\begin{align*}
    \text{Pr}[E_2] \geq \text{Pr}[E_2|E_1]\text{Pr}[E_1] &\geq (1 - \frac{1}{r} - \frac{12}{r^2})(1 - \frac{1}{\sqrt{2r}} - \frac{1}{r})(1 - \frac{1}{r})(1 - (\frac{5}{4r} + \frac{1}{4C^2_{CW}r^{\frac{3}{2}}})) \\
    &\geq \Big(1 - \frac{13}{r^2}\Big)^2\Big(1 - (\frac{5}{4r} + \frac{1}{4C^2_{CW}r^{\frac{3}{2}}})\Big).
\end{align*}
\textbf{\underline{Complexity Analysis:}} We analyse the steps of the algorithm and deduce the number of arithmetic operations required for its execution.
\begin{enumerate}
    \item Step 1 of the algorithm runs Algorithm \ref{algo:robustSUBrecovery} to return a solution $P$ to approximate SUB recovery problem for the input tensor $T$. By Theorem \ref{thm:SUBrecoverymaintheorem}, this requires $O(n^3)$ arithmetic operations.
    \item Step 2 of the algorithm uses Algorithm \ref{algo:changeofvariablesexact} to perform a change of basis.
    By Theorem \ref{thm:fastchangeofvarsproof}, this requires $O(n^{\omega(2)})$ arithmetic operations.
    \item Step 3 of the algorithm uses Algorithm \ref{algo:Jennrich} to compute a decomposition $v_1,...,v_r$ of the tensor $S \in \C^r \otimes \C^r \otimes \C^r$. Using Theorem \ref{thm:completetensordecomposition}, this requires $O(r^3 + T_{MM}(r)\log^2(\frac{rB}{\varepsilon}))$ arithmetic operations.
    \item Step 4 of algorithm computes $u_i = Pv_i$ for all $i \in [r]$ and this can be computed with $O(nr^2)$ arithmetic operations.
\end{enumerate}
\end{proof}

\subsection{Robust Linear time algorithm}\label{sec:undercompletelineartime}
{As discussed in Section \ref{sec:inmoredetail}, in this section  we ``open the box"   of the robust algorithm for  decomposition of diagonalisable tensors from \cite{KoiranSaha22}. 
We show that with a few modifications, we can obtain a robust linear time algorithm for  decomposition of $r$-diagonalisable tensors.
\par
We first present a version of Algorithm \ref{algo:undercompleteexact} after opening the blackbox of Algorithm~\ref{algo:Jennrich}  without any modifications.}

\begin{algorithm}[H] \label{algo:undercompleteexactbnonlinear}
\SetAlgoLined
\nonl In the following algorithm, ${C} ,C_{\text{gap}},C_{\eta} > 0$ and $c_F > 1$ are some absolute constants that are set in \cite{KoiranSaha22}. \\
\nonl \textbf{Input:}  A degree-3 $r$-diagonalisable symmetric tensor $T \in \C^n \otimes \C^n \otimes \C^n$, an estimate $B$ for the condition number of the tensor $T$ and an accuracy parameter $\varepsilon (< 1)$.  \\
\nonl \textbf{Output:} linearly independent vectors $l_1,...,l_r$ such that $T = \sum_{i=1}^r l_i^{\otimes 3}$ \\
Run Algorithm \ref{algo:robustSUBrecovery} on input $(T,B,\varepsilon)$ and let $P { \in \C^{n \times r}}$ be the output. \\
\nonl Pick $ a_1,...,a_r, b_1,...,b_r$ uniformly and independently from the finite grid $ G^{2r}_{\eta}$ where $G_{\eta}$ is defined in (\ref{eq:discretegrid}) and $\eta = \frac{1}{C_{\eta}r^{\frac{17}{2}}B^4}$ is the grid size. \\
Compute $S = (\overline{P} \otimes \overline{P} \otimes \overline{P}).T { \in \C^{r \times r \times r}}$ and let $S_1,...,S_r \in \C^{r \times r}$ be the slices of $S$.  \\
{ \nonl Set $\varepsilon_0 = \frac{\varepsilon}{2\sqrt{r}}$.} \\
Compute $S^{(a)} = \sum_{i=1}^r a_iS_i { \in \C^{r \times r}}$ and $S^{(b)} = \sum_{i=1}^r b_iS_i { \in \C^{r \times r}}$. \\
Compute $S^{(a)'} = (S^{(a)})^{-1}$ \\
Compute $D = S^{(a)'}S^{(b)}$ \\
\nonl Set $k_{\text{gap}} := \frac{1}{C_{\text{gap}}r^6B^3}$, $k_{F} := c_Fr^5B^3$ and { $\varepsilon_1 = \frac{\varepsilon_0^3}{Cr^{12}B^{\frac{9}{2}}}$}. \\
Let $v_1,...,v_r$ be the output of EIG-FWD on the input $(D,\varepsilon_1, \frac{3rB}{k_{\text{gap}}}, 2B^{\frac{3}{2}}\sqrt{rk_F})$ where EIG-FWD is the numerically stable algorithm for approximate matrix diagonalisation from \cite{9317903},\cite{KoiranSaha22}.\\
Compute $W = V^{-1}$ where $V \in \C^{r \times r}$ is the matrix with columns $v_1,...,v_r$ and let $w_1,...,w_r$ be the rows of $W$. \\
Compute $\alpha_1,...,\alpha_r = TSCB(V,S)$ where TSCB is the algorithm for computing the trace of slices of a tensor after change of basis as described in Algorithm \ref{algo:fastcob}. \footnote{{ Note here we just need the square matrix version of the Algorithm \ref{algo:fastcob} which can already be found in \cite{KoiranSaha22}}}\\
Compute $z_i = \alpha_i^{\frac{1}{3}}w_i$ for all $ i \in [r]$.  \\
Compute $l_i = P z_i$ for all $i \in [r]$. \\
\nonl Output $l_1,...,l_r$
\caption{Expanded version of Algorithm \ref{algo:undercompleteexact}}
\end{algorithm}

The above algorithm is the expanded version of Algorithm \ref{algo:undercompleteexact} which we include here for a better exposition. Notice here that Steps 1,2 and 4 of Algorithm \ref{algo:undercompleteexact} coincide with Steps 1,2 and 10 Algorithm \ref{algo:undercompleteexactbnonlinear} respectively. Only Step 3 of Algorithm \ref{algo:undercompleteexact} has been replaced by the corresponding implementation of Algorithm \ref{algo:Jennrich} on the reduced tensor $S \in \C^r \otimes \C^r \otimes \C^r$ and this accounts for Steps 3-9 of Algorithm \ref{algo:undercompleteexactbnonlinear}.
\par
As we have seen in Theorem \ref{thm:undercompleteexactproof}, the above algorithm is not linear time owing to the computation of the tensor $S$ in Step 2. The goal is to remove the explicit computation of the tensor $S$ in order to obtain a linear time algorithm. To do this, we take the following strategy:
\begin{itemize}
    \item We avoid Step 2 of the algorithm and move to Step 3 directly. That is, instead of computing the tensor $S$ explicitly, we compute the random linear combination of the slices of $S$, denoted by $S^{(a)}$ and $S^{(b)}$. This can be done by first picking $a$ and $b$ at random from the discrete grid $G^r_{\eta}$ and then, using Algorithm \ref{algo:changeofvariableslincomb} on inputs $(T,\overline{P},a)$ and $(T,\overline{P},b)$ respectively. 
    \item {By looking at Algorithm \ref{algo:undercompleteexactbnonlinear}, one can observe that }the only other step that depends on the computation of the tensor $S$ explicitly is Step 8. Let $V$ be the matrix with normalized eigenvectors returned in Step 6 of the algorithm. Then, in Step 8, we want to compute the scaling coefficients $\alpha_1,...,\alpha_r = TSCB(V,S) $. Then for all $i \in [r]$, $\alpha_i = Tr(T'_i)$  where $T' = (V \otimes V \otimes V).S$ and $T'_1,...,T'_r$ are the slices of $T'$.
    \par
    {We cannot perform this operation, since we cannot compute $S$ directly (this is an expensive step).
    Instead we replace this by the following steps:
    \begin{itemize}
        \item Compute $M = \overline{P}.V$.
        \item Compute $\alpha_1,...,\alpha_r = TSCB(M,T)$ where $TSCB$ is the Algorithm for computing {the trace of slices after a change of basis.
        We will therefore still  call TSCB like in Step 8 of Algorithm~\ref{algo:undercompleteexactbnonlinear}, but on input $(M,T)$ instead of $(V,S)$.}
    \end{itemize}}
\end{itemize}
Using the above-mentioned strategy, the following is the proposed robust linear time algorithm for undercomplete tensor decomposition.

\begin{algorithm}[H] \label{algo:undercompleteexactlinear}
\SetAlgoLined
\nonl In the following algorithm, ${ C} ,C_{\text{gap}},C_{\eta} > 0$ and $c_F > 1$ are some absolute constants that are set in \cite{KoiranSaha22}. \\
\nonl \textbf{Input:} A degree-3 $r$-diagonalisable symmetric tensor $T \in \C^n \otimes \C^n \otimes \C^n$, an estimate $B$ for the condition number of the tensor $T$ and an accuracy parameter $\varepsilon (< 1)$.  \\
\nonl \textbf{Output:} linearly independent vectors $l_1,...,l_r$ such that $T = \sum_{i=1}^r l_i^{\otimes 3}$ \\
Run Algorithm \ref{algo:robustSUBrecovery} on input $(T,B,\varepsilon)$ and let $P$ be the output. \\
\nonl Pick $ a_1,...,a_r, b_1,...,b_r$ uniformly and independently from the finite grid $ G^{2r}_{\eta}$ where $G_{\eta}$ is defined in (\ref{eq:discretegrid}) and $\eta = \frac{1}{C_{\eta}r^{\frac{17}{2}}B^4}$ is the grid size. \\
Compute $S^{(a)} = LSCSB(T,\overline{P},a)$ and $S^{(b)} = LSCSB(T,\overline{P},b)$ where $\text{LSCSB}$ denotes Algorithm \ref{algo:changeofvariableslincomb} (it computes a linear combination of the slices of the tensor after a change of basis). \\
{ \nonl Set $\varepsilon_0 = \frac{\varepsilon}{2\sqrt{r}}$.} \\
Compute $S^{(a)'} = (S^{(a)})^{-1}$ \\
Compute $D = S^{(a)'}S^{(b)}$ \\
\nonl Set $k_{\text{gap}} := \frac{1}{C_{\text{gap}}r^6B^3}$, $k_{F} := c_Fr^5B^3$ and {$\varepsilon_1 = \frac{\varepsilon_0^3}{Cr^{12}B^{\frac{9}{2}}}$}.  \\
Let $v_1,...,v_r$ be the output of EIG-FWD on the input $(D,\varepsilon_1, \frac{3rB}{k_{\text{gap}}}, 2B^{\frac{3}{2}}\sqrt{rk_F})$ where EIG-FWD is the numerically stable algorithm for approximate matrix diagonalisation from \cite{9317903},\cite{KoiranSaha22}.\\
Compute $W = V^{-1}$ and let $w_1,...,w_r$ be the rows of $W$. \\
Let $M = \overline{P}V$ where $V$ is the matrix with columns $v_1,...,v_r$ \\
Compute $\alpha_1,...,\alpha_r = TSCB(M,T)$ where TSCB is the algorithm for computing the trace of slices of a tensor after change of basis as described in Algorithm \ref{algo:fastcob}. \\
Compute $z_i = \alpha_i^{\frac{1}{3}}w_i$ for all $ i \in [r]$.  \\
Compute $l_i = P z_i$ for all $i \in [r]$. \\
\nonl Output $l_1,...,l_r$
\caption{Undercomplete decomposition of symmetric tensors in linear time}
\end{algorithm}
We restate the theorem here for completeness.
\newline
\textbf{Theorem }\ref{thm:undercompleteexactprooflinear}: Let $T \in \C^n \otimes \C^n \otimes \C^n$ be an $r$-diagonalisable tensor for some $r \leq n$ and let $T' \in \C^n \otimes \C^n \otimes \C^n$ be such that $||T - T'|| \leq \delta \in~(0, \frac{1}{\text{poly}(n,B)r^{C\log^4(\frac{rB}{\varepsilon})}})$ for some constant $C$.\footnote{The exact values for $\text{poly}(n,B)$ and $C$ are set in (\ref{eq:delta}) and (\ref{eq:eps}) respectively.} Then on input $T'$, a desired accuracy parameter $\varepsilon$ and some estimate $B \geq \kappa(T)$, Algorithm \ref{algo:undercompleteexactlinear} outputs an $\varepsilon$-approximate solution to the tensor decomposition problem for $T$ with probability at least
$$\Big(1 - \frac{13}{r^2}\Big)^2\Big(1 - (\frac{5}{4r} + \frac{1}{4C^2_{CW}r^{\frac{3}{2}}})\Big). \footnote{Here $C_{CW}$ is the constant from the anti-concentration inequalities due to \cite{CW01}}$$
The algorithm requires $O(n^3 + T_{MM}(r)\log^2(\frac{rB}{\varepsilon}))$ arithmetic operations.
\begin{proof}
{To justify the correctness and robustness of the algorithm, we show that { on any input and for any choice of the random bits $a_1,\ldots,a_r,b_1,\ldots,b_r$,} the algorithm computes the same quantities as Algorithm~\ref{algo:undercompleteexactbnonlinear}. Hence the correctness and robustness of Algorithm \ref{algo:undercompleteexactbnonlinear} from Theorem \ref{thm:undercompleteexactproof} will automatically give us the same correctness and robustness for Algorithm \ref{algo:undercompleteexactlinear}.
\par
As discussed before Algorithm \ref{algo:undercompleteexactlinear}, it differs from Algorithm \ref{algo:undercompleteexactbnonlinear} in two steps:
\begin{itemize}
    \item Let $P$ be the output of Step 1 of both  algorithms. Then Step 2 of Algorithm \ref{algo:undercompleteexactlinear} picks vectors $a_1,...,a_r,b_1,...,b_r$ uniformly at random from the discrete grid and computes $S^{(a)} = \sum_{i=1}^r a_i$ and $S^{(b)}$ where $S_i$ are the slices of the tensor $S = (\overline{P} \otimes \overline{P} \otimes \overline{P}).T$. This is exactly the same output as computed at the end of Step 3 of Algorithm \ref{algo:undercompleteexactbnonlinear}. It avoids the computation of the tensor $S$ explicitly.
    \item  The only step that depends on the explicit computation of the tensor $S$ in Algorithm \ref{algo:undercompleteexactbnonlinear} is Step 8. This has been replaced with Steps 7 and 8 in Algorithm \ref{algo:undercompleteexactlinear}. The rest of the steps can be executed with just the computation of $S^{(a)}$ and $S^{(b)}$. We claim that steps 7 and 8 in Algorithm \ref{algo:undercompleteexactlinear} compute the exact same quantity as is computed by step 8 of Algorithm \ref{algo:undercompleteexactbnonlinear}.
    \par

{ Suppose therefore that Algorithms 9 and 10 are given the same (possibly perturbed) input $T$, and that the same random bits  are used by both algorithms. In this case,
\begin{itemize} 
\item[(i)] Step 8 of Algorithm \ref{algo:undercompleteexactbnonlinear} 
computes the trace of the slices of
$(V \otimes V \otimes V).S$, where  $S = (\overline{P} \otimes \overline{P} \otimes \overline{P}).T$.
\item[(ii)] Step 8 of Algorithm \ref{algo:undercompleteexactlinear} computes 
the trace of the slices of 
$(M \otimes M \otimes M).T$ 
where $M=\overline{P}V$.
\end{itemize}
Note that the same matrices $P$ and $V$ appear in (i) and (ii). This is because:  $P$ is computed at the first step of both algorithms, and these two steps are identical; $V$ is computed at Step 6 of Algorithm~\ref{algo:undercompleteexactbnonlinear}
and Step 5 of Algorithm~\ref{algo:undercompleteexactbnonlinear},
and we have shown that the algorithms are equivalent up to these points.  Since $M=\overline{P}V$,
$$(V \otimes V \otimes V).S=
(M \otimes M \otimes M).T$$ by~(\ref{eq:2changes}),
so the two algorithms compute the trace of slices of the same tensor in their $8^{th}$ steps.}
 \end{itemize}
 From this analysis we can conclude that the two algorithms produce exactly the same outputs,
and Algorithm~\ref{algo:undercompleteexactlinear} therefore inherits the correctness and robustness properties of Algorithm~\ref{algo:undercompleteexactbnonlinear}.
}
\par
\textbf{Complexity Analysis of Algorithm \ref{algo:undercompleteexactlinear}:}
We analyse the steps of the algorithm and deduce the number of arithmetic operations required to execute it.
\begin{enumerate}
    \item By Theorem \ref{thm:SUBrecoverymaintheorem},  Step 1 of the algorithm can be performed in $O(n^3)$ arithmetic operations.
    \item Step 2 of the algorithm requires two applications of Algorithm \ref{algo:changeofvariableslincomb} on a tensor. Using Theorem \ref{thm:changeofvarslincomb}, this can be performed using $O(n^3)$ arithmetic operations.
    \item Steps 3,4,5 and 6 are the same as Algorithm \ref{algo:Jennrich} on matrices in $M_r(\C)$. So using Theorem \ref{thm:completetensordecomposition}, these steps can be performed in $O(r^3 + T_{MM}(r)\log^2(\frac{rB}{\varepsilon}))$ arithmetic operations. 
    \item Step 7  performs a matrix multiplication between $\overline{P} \in \C^{n \times r}$ and $V \in M_r(\C)$. This can be done in $O(nr^2)$ arithmetic operations.
    \item Step 8 runs Algorithm \ref{algo:fastcob} on inputs $T \in \C^n \otimes \C^n \otimes \C^n$ and matrix $M \in~\C^{n \times r}$. Using Theorem \ref{thm:fastcob}, this can be performed in $O(n^3)$ arithmetic operations.
    \item Step 9 multiplies vector $w_i \in \C^r$ by a scalar for all $i \in [n]$ and this can be performed in $O(r^2)$ arithmetic operations.
    \item Step 10 computes the matrix-vector product between the $n \times r$ matrix $P$ and the vectors $z_i \in \C^r$ for all $i \in [r]$. This can be done in $O(nr^2)$ arithmetic operations.
\end{enumerate}
So the total number of arithmetic operations required for this algorithm is $O(n^3 + T_{MM}(n)\log^2(\frac{nB}{\varepsilon}))$.    
\end{proof}
\section{Analysis of Condition Numbers}
\label{sec:smoothedanalysis}

Recall from Definition~\ref{def:conditionnumberundercomplete} 
that the tensor condition condition number $\kappa(T)$ is defined
as $\kappa_F(U) = ||U||_F^2 + ||U^{\dagger}||_F^2$ where
the rows of $U$ are the vectors occurring in the decomposition
of $T$.
In this section we present a smoothed analysis of the matrix condition number $\kappa_F(U)$. For this we rely heavily 
on the book by Bürgisser and Cucker~\cite{BC13}. 
We first recall some of their notations.

\subsection{Norms}
For any $x \in \R^n$, we define the following norms
\begin{equation}
\begin{split}
    ||x||_r &:= \Big(\sum_{i=1}^n |x_i|^r\Big)^{\frac{1}{r}} \text{ for any real number } r \geq 1 \\
    ||x||_{\infty} &:= \max_{i \in [n]} |x_i|.    
\end{split}
\end{equation}
For any matrix $A$, we  define the following norms
\begin{equation}
\begin{split}
        ||A||_{rs} &:= \sup_{x \in \R^p, x \neq 0} \frac{||Ax||_s}{||x||_r} = \sup_{||x||_r = 1} ||Ax||_s \\
        ||A||_F &:= \sqrt{\sum_{i,j=1}^n |A_{ij}|^2}
\end{split}
\end{equation}
\subsubsection{Condition Numbers}
Let $m = n$ and fix norms $||.||_r$ and $||.||_s$ on $\R^n$. Let $\Sigma := \{A \in \R^{n \times n} | \text{det}(A) = 0\}$ denote the set of \textit{ill-posed} matrices. Define $\mathcal{D} := \R^{n \times n} \setminus \Sigma$. Define the map $\kappa_{rs} : \mathcal{D} \xrightarrow{} \R$ by 
\begin{equation}
    \kappa_{rs}(A) := ||A||_{rs}||A^{-1}||_{sr}.
\end{equation}
This is referred to as the normwise condition number for linear equation solving. 
Recall that the Frobenius condition number of a matrix $A \in \R^{m \times n}$ is defined as
\begin{equation}
    \kappa_F(A) := ||A||_F^2 + ||A^{\dagger}||_F^2.
\end{equation}
Note that if $A \in \mathcal{D}$, $\kappa_F(A) = ||A||_F^2 + ||A^{-1}||_F^2$.

\subsection{Smoothed Analysis of Condition Number for the Undercomplete Case}
Recall that a symmetric order-$3$ tensor $T \in \R^n \otimes \R^n \otimes \R^n$ is $r$-diagonalisable if there exists an $r \times n$ matrix $U$ where $r \leq n$ with linearly independent rows $u_1,...,u_r \in \R^n$ such that $T = \sum_{i=1}^r u_i^{\otimes 3}$.
\par
\textbf{Smoothed Analysis Model: }Let $\overline{T} \in (\R^n)^{\otimes 3}$ be a symmetric tensor of rank $r \leq n$. Then $\overline{T} = \sum_{i=1}^r \overline{u_i}^{\otimes 3}$. Let $u_i \sim \mathcal{N}(\overline{u}_i, \sigma^2 I_n)$ for all $i \in [r]$ and define $T = \sum_{i=1}^r u_i^{\otimes 3}$.
Recall from Definition \ref{def:conditionnumberundercomplete} that the condition number of $T$ is defined as $\kappa(T) := \kappa_F(U) = ||U||_F^2 + ||U^{\dagger}||_F^2$ where $U$ is the matrix with rows $u_i$.
\begin{theorem}\label{thm:lipcontfunction}[Theorem 4.4 in \cite{BC13}]
    Let $f: \R^n \xrightarrow{} \R$ be an almost everywhere differentiable and Lipschitz continuous function with Lipschitz constant $L$. Then we have, for all $t \geq 0$ and $x \in \R^n$ drawn from the standard Gaussian distribution, that
    \begin{align*}
        \text{Pr}_{x \in \mathcal{N}(0,I_n)}\Big[f(x) \geq \mathbb{E}(f) + t\Big] \leq e^{-\frac{2}{\pi^2L^2}t^2}.
    \end{align*}
\end{theorem} 
For a standard Gaussian $X \in \R^{m \times n}$, define 
\begin{equation}
    Q(m,n) := \frac{1}{\sqrt{n}}\mathbb{E}\Big(||X||\Big)
\end{equation}
\begin{lemma}[Lemma 4.14 in \cite{BC13}]\label{lem:Qmnbounds}
For standard Gaussian matrices, $A \in \R^{m \times n}$, we have 
\begin{equation}
    \sqrt{\frac{n}{n+1}} \leq Q(m,n) \leq 6.
\end{equation}
\end{lemma}
\begin{corollary}\label{corr:normAbounds}
Let $\overline{A} \in \R^{m \times n}$ and $\sigma \in (0,1]$.
Then for all $t > 0$,
\begin{align*}
    \text{Pr}_{A \sim \mathcal{N}(\overline{A},\sigma^2 I)} \Big[||A|| \geq 6\sigma \sqrt{n}+ \sigma t + ||\overline{A}|| \Big] \leq e^{-\frac{2t^2}{\pi^2}}.
\end{align*}
\end{corollary}
\begin{proof}
Let $f: \R^{m \times n} \xrightarrow{} \R$ be defined as $f(X) = ||X||$ where $||.||$ is the spectral norm of a matrix. By the triangle inequality, for two matrices $X_1,X_2 \in  \R^{m \times n}$, $\Big|||X_1|| - ||X_2||\Big| \leq ||X_1 - X_2||$. Hence $f$ is Lipschitz-continuous with Lipschitz constant $1$. By Theorem \ref{thm:lipcontfunction},
for all $t > 0$,
\begin{equation}\label{eq:centred}
    \text{Pr}_{X_{ij} \in \mathcal{N}(0,1)} \Big[||X|| \geq Q(m,n)\sqrt{n} + t\Big] \leq e^{-\frac{2t^2}{\pi^2}}.
\end{equation}
If $||A|| \geq Q(m,n)\sqrt{n} + \sigma t + ||\overline{A}||$, then $||\frac{A - \overline{A}}{\sigma}|| \geq \frac{1}{\sigma}(||A|| - ||\overline{A}||) \geq Q(m,n) \sqrt{n} +~t $. Moreover if $A \in \R^{m \times n} \sim \mathcal{N}(A,\sigma^2 I)$, then it follows that $X:= \frac{A - \overline{A}}{\sigma}$ is standard Gaussian in $\R^n$. We can now conclude using (\ref{eq:centred}) and Lemma \ref{lem:Qmnbounds}. 
\end{proof}
\begin{theorem}\label{thm:normApseudoinv}[Proposition 4.19 in \cite{BC13}]
    Let $\overline{A} \in \R^{m \times n}$, $\sigma > 0$ and put $\lambda := \frac{m-1}{n}$. Then for random $A \sim \mathcal{N}(\overline{A},\sigma^2 I)$, we have that for any $t > 0$,
    \begin{equation}
        \text{Pr}_{A \sim \mathcal{N}(\overline{A},\sigma^2 I)} \Big[||A^{\dagger}|| \geq \frac{t}{1-\lambda}\Big] \leq c(\lambda)\Big(\frac{e}{\sigma\sqrt{n}t}\Big)^{(1-\lambda)n}.
    \end{equation}
    where $c(\lambda) := \sqrt{\frac{1+\lambda}{2(1- \lambda)}}$.
\end{theorem}
The following is the main theorem of this section.
 \begin{theorem} \label{th:smooth}
Let $\overline{A}  = (\overline{a}_{ij})$ be an arbitrary $n \times m$ matrix. Let $A$ be a matrix such that for all $i,j \in~[n]$, its $(i,j)$-th entry is sampled independently at random from the Gaussian distribution centred at $\overline{a}_{ij}$ and variance $\sigma^2$. More formally, $A \sim \mathcal{N}(\overline{A},\sigma^2 I)$. Then $$\kappa_F(A) \leq 98\sigma^2 n^3 + 2n ||\overline{A}||^2 + \frac{e^2 n^4}{\sigma^2}$$ with probability at least $1 - (\frac{1}{\sqrt{n}} + \frac{1}{e^{\frac{2n}{\pi^2}}})$.   
\end{theorem}
\begin{proof}
Since $m \leq n$,
we have $1- \lambda \geq \frac{1}{n}$ and $1 + \lambda \leq 2$. This gives us that $c(\lambda) \leq \sqrt{n}$. Using Theorem \ref{thm:normApseudoinv} for $t = \frac{e\sqrt{n}}{\sigma}$ and the fact that $||A||_F \leq ||A||\sqrt{n}$ for any matrix $A$,
\begin{equation}\label{eq:normAinv}
        \text{Pr}_{A \sim \mathcal{N}(\overline{A},\sigma^2 I)} \Big[||A^{\dagger}||_F \geq \frac{en}{\sigma(1-\lambda)}\Big] \leq c(\lambda)\Big(\frac{1}{n}\Big)^{(1-\lambda)n} < \frac{1}{\sqrt{n}}.
\end{equation}
By Corollary \ref{corr:normAbounds} for $t = \sqrt{n}$, we also have
\begin{equation}\label{eq:normAbounds}
     \text{Pr}_{A \sim \mathcal{N}(\overline{A},\sigma^2 I)} \Big[||A||_F \geq \sqrt{n}(7\sigma \sqrt{n} + ||\overline{A}||) \Big] \leq e^{-\frac{2n}{\pi^2}}.
\end{equation}
Combining (\ref{eq:normAinv}) and (\ref{eq:normAbounds}) using the union bound, if $A \sim \mathcal{N}(\overline{A},\sigma^2 I)$ then 
\begin{align*}
\kappa_F(A) &= ||A||_F^2 + ||A^{\dagger}||_F^2 \\
&\leq n(7\sigma\sqrt{n} + ||\overline{A}||)^2 + \frac{e^2n^2}{\sigma^2(1-\lambda)^2} \leq 2n(49\sigma^2 n + ||\overline{A}||^2) + \frac{e^2n^4}{\sigma^2}
\end{align*}
with probability at least $ 1 - \Big(\frac{1}{\sqrt{n}} + \frac{1}{e^{\frac{2n}{\pi^2}}}\Big)$. The last inequality follows from Cauchy-Schwarz inequality and  the bound $\frac{1}{1- \lambda} \leq n$. This gives us the desired result.
\end{proof}
In Appendix~\ref{app:simpleavgcasediag} we present a more elementary analysis for the
special case where $r=n$ and $\overline{A}=0$ in Theorem~\ref{th:smooth}.
This is also based on results from~\cite{BC13}.
\par
This previous result automatically gives us the following result about the condition number of tensors.
\begin{theorem}
Let $T' \in \R^n \otimes \R^n \otimes \R^n$ be a symmetric tensor of rank $r\leq~n$ with a decomposition $T' = \sum_{i=1}^r (u'_i)^{\otimes 3}$ into rank-one components where $u'_i \in~\R^n$. We denote by $U'$ the matrix with rows $u'_i$. Let $U \sim~\mathcal{N}(U',\sigma^2 I)$ and let $T$ be a symmetric tensor defined as $T = \sum_{i=1}^r u_i^{\otimes 3}$ where the $u_i$ are the rows of $U$. Then 
\begin{align*}
    \kappa(T) \leq 98\sigma^2 n^3 + 2n||U'||_F^2 + \frac{e^2n^4}{\sigma^2}
\end{align*}
with probability at least $ 1 - \Big(\frac{1}{\sqrt{n}} + \frac{1}{e^{\frac{2n}{\pi^2}}}\Big)$. 
\end{theorem}
\begin{proof}
This follows from Theorem \ref{th:smooth} and the definition of the condition number of tensor decomposition.
\end{proof}

\section*{Acknowledgement}
Subhayan Saha acknowledges the support by the European Union (ERC consolidator, eLinoR, no 101085607). 

%\bibnote{KoiranSaha22}{Full version on \href{https://arxiv.org/abs/2211.07407}{arXiv}}
%\bibliographystyle{alpha}
%\bibliography{sample}
\appendix

\section{Appendix: Explicit change of basis}\label{appsec:cobexact}
The columns of a $n\times n^2$ matrix $M$ can be ordered in  lexicographic order as pairs $(i,j)$ where $i,j \in [n]$.  We denote by $M_{i,jk} = M_{i,(j,k)}$  the entry in the $i$-th row and  $(j,k)$-th column.

\begin{algorithm}[H] \label{algo:changeofvariablesexact}
\SetAlgoLined
\nonl \textbf{Input:} An order-$3$ symmetric tensor $T \in \mathbb{C}^{n \times n \times n}$, a matrix $V = (v_{ij}) \in \mathbb{C}^{n \times n}$.\\
Let $T_1,...,T_n$ be the slices of $T$. Let $T = \begin{bmatrix}
T_1 & T_2 & \hdots &T_n.\\
\end{bmatrix}$  be an $n \times n^2$ matrix.
{ By abuse of language, we use the same notation $T$ for the input tensor and its $n\times n^2$ flattening.} 
Note that 
{ $T_{i,jk} = (T_j)_{i,k}$} \\
Let $D_i$ be an $n \times n$ matrix such that { $(D_i)_{jk} = (V^T T)_{i,jk}$}. \\
Let $D = \begin{bmatrix}
D_1 & D_2 & \hdots &D_n\\
\end{bmatrix}$ be a $n \times n^2$ matrix and let $M = V^T D$. \\
Let  $M'$ be a $n^2 \times n$ matrix such that  $M'_{ij,k} = M_{i,jk}$ \\
Let $S$ be an $n \times n \times n$ array such that $S_{i,j,k} = (M'V)_{ij,k}$ \\
Output $S$
\caption{Algorithm for change of basis ($\text{CB}$) in a symmetric tensor.}
\end{algorithm}
\begin{theorem}\cite{GU'18}\label{thm:rectmatmul}
The product of an $n \times n^2$ matrix and an $n^2 \times n$ matrix can be computed in $O(n^{\omega(2)})$ where $\omega(2) = 3.251640$.
\end{theorem}
\begin{lemma}\label{lem:fastchangeofvar}
Let $T \in \C^n \otimes \C^n \otimes \C^n$ be a symmetric tensor with slices $T_1,...,T_n$. Let $V$ be an $n \times r$ matrix over $\C$ and let $S = (V \otimes V \otimes V).T \in \C^r \otimes \C^r \otimes \C^r$. Then the slices $S_1,...,S_r \in M_r(\C)$ of $S$ are given by
$$S_i= V^TD_iV \text{ where } (D_i)_{jk} = (V^T T_j)_{ik} \text{ for all } i\in [r].$$
\end{lemma}
\begin{proof}
Let $T$ be a symmetric tensor and $T_1,...,T_n$ be the slices of $T$. Let $S = (V \otimes V \otimes V).T$ and let $S_1,...,S_n$ be the slices of $S$. By Theorem \ref{thm:P3structural},
\begin{equation}
    S_i = V^TD_iV  \text{ where } D_i = \sum_{m=1}^n v_{mi}T_m.
\end{equation}
Since $T$ is a symmetric tensor,
$(T_m)_{jk} = (T_j)_{mk}$.
Therefore,
\begin{align*}
    (D_i)_{jk} = \sum_{m=1}^n v_{mi}(T_m)_{jk} &= \sum_{m=1}^n v_{mi}(T_j)_{mk} \\
    &= \sum_{m=1}^n (V^T)_{im}(T_j)_{mk} = (V^T T_j)_{ik}.
\end{align*}
\end{proof}
\begin{theorem}\label{thm:fastchangeofvarsproof}
Let $T \in \C^n \otimes \C^n \otimes \C^n$ be an order-$3$ symmetric tensor and $V$ be an $n \times r$ matrix. Then Algorithm \ref{algo:changeofvariablesexact} returns the tensor $S$ in $\C^r \otimes \C^r \otimes \C^r$ such that 
$S = (V \otimes V \otimes V). T$. 
The algorithm requires $O(n^{\omega(2)})$ many arithmetic operations where $\omega(2) = 3.251640$.
\end{theorem}
\begin{proof}
We first show that the algorithm always returns the correct output. More formally, we want to show that if Algorithm \ref{algo:changeofvariablesexact} is run on inputs $T,V$, then the tensor returned at the end of the algorithm is indeed equal to $(V \otimes V \otimes V).T$.
\par
Following Algorithm \ref{algo:changeofvariablesexact} and using block multiplication, at the end of Step~4 we have { $$(D_i)_{jk} = (V^T T)_{i,jk} = \sum_{m=1}^n v_{mi}T_{m,jk} = \sum_{m=1}^n v_{mi}(T_j)_{m,k} = (V^T T_j)_{i,k}.$$} We have defined $D = \begin{bmatrix}
D_1 & D_2 & \hdots &D_n\\
\end{bmatrix}$. This gives us that $M = \begin{bmatrix}
V^T D_1 & V^T D_2 & \hdots &V^T D_n\\
\end{bmatrix}$. Since $M'_{ij,k} = M_{i,jk}$ , we also have $M' ~=~\begin{bmatrix}
&V^T D_1 \\
& V^T D_2 \\
& \vdots \\
&V^T D_n\\
\end{bmatrix}$. Using block multiplication rules again,
$M'V = \begin{bmatrix}
&V^T D_1V \\
& V^T D_2V \\
& \vdots \\
&V^T D_nV \\
\end{bmatrix}$. Let $S$ be the {$r \times r \times r$} tensor returned by the algorithm. Then ${ S}_{i,j,k} = (V^T D_i V)_{j,k}$ where { $(D_i)_{j,k} = (V^T T_j)_{i,k}$. Hence $S=(V \otimes V \otimes V). T$ by Lemma \ref{lem:fastchangeofvar}.} 
\par
Now we analyse the steps of the algorithm and compute the number of arithmetic operations required:
\begin{itemize}
    \item In Step 1, we construct the matrix $T$ by accessing every entry of the tensor $T$. This requires $n^3$ many arithmetic operations.
    \item In Step 2, we compute $V^T T$ where $V \in \C^{n \times r}$ and $T \in \C^{n \times n^2}$. Using Theorem \ref{thm:rectmatmul}, this can be done in $O(n^{\omega(2)})$ many steps. This gives us all the entries for $D_i$.
    \item In Step 3, we create a $n \times n^2$ matrix $D$ using the $D_i$'s and compute $M = V^TD$. Again using Theorem \ref{thm:rectmatmul}, this can be done in $O(n^{\omega(2)})$ many steps.
    \item Computing $n^2 \times n$ matrix $M'$ just requires accessing every entry of $M$ and this can be done in $n^3$ many steps.
    \item Finally, we compute $M'V$ which can again be done in $O(n^{\omega(2)})$ many steps using Theorem \ref{thm:rectmatmul}.
\end{itemize}
So the total number of arithmetic operations required by this algorithm is $O(n^{\omega(2)})$.
\end{proof}
\section{A Simple Analysis of the Average Case for Diagonalisable Tensors}\label{app:simpleavgcasediag}
We denote by $\mathbb{S}^{n-1}$ the unit sphere $ \Big\{x \in \R^n \Big| ||x|| = 1\Big\}$. 
\begin{theorem}\label{thm:kappabounds}[Theorem 2.45 in \cite{BC13}]
For any $n > 2$ and $0 < \varepsilon \leq 1$,
\begin{equation}
    \text{Pr}\Big[\kappa_{2\infty}(A) \geq \frac{1}{\varepsilon}\Big] \leq \sqrt{\frac{2}{\pi}}n^{\frac{5}{2}}\varepsilon.
\end{equation}
when $A$ is picked uniformly at random from the uniform distribution on $(\mathbb{S}^{n-1})^n$.
\end{theorem}
A function $g: \R^n \xrightarrow{} \R$ is \textit{scale-invariant} when for all $a \in \R^n$ and for all $\lambda > 0$, $g(\lambda a) = g(a)$.
\begin{lemma}\label{lem:normalspherical}[Corollary 2.13 in \cite{BC13}]
Let $g: \R^n \xrightarrow{} \R$ be a scale-invariant, integrable function and denote by $g|_{\mathbb{S}^{n-1}}$ its restriction to $\mathbb{S}^{n-1}$. Then we have for all $t \in \R$,
\begin{align*}
    \text{Pr}_{x \sim \mathcal{N}(0,I_n)}\Big[g(x) \geq t\Big] = \text{Pr}_{x \sim U(\mathbb{S}^{n-1})}\Big[g|_{\mathbb{S}^{n-1}} (x) \geq t\Big]
\end{align*}
\end{lemma}
\textbf{Remark: }A function $g : \R^{n_1 \times ... \times n_k} \xrightarrow{} \R$ is called \textit{scale-invariant by blocks} when $g(\lambda_1a_1,...,\lambda_ka_k) = g(a_1,...,a_k)$ for all $\lambda_1,...,\lambda_k > 0$. An extension of Lemma \ref{lem:normalspherical} is the following:
\begin{corollary}\label{corr:normalspherical}[Remark 2.24 in \cite{BC13}]
Let $g: \R^n \xrightarrow{} \R$ be an integrable function that is scale-invariant by blocks. Then for all $t \in \R$,
\begin{align*}
    \text{Pr}_{x \sim \mathcal{N}(0,I_n)}\Big[g(x) \geq t\Big] = \text{Pr}_{x \sim U(\mathbb{S}^{n_1-1}) \times ... \times U(\mathbb{S}^{n_k-1})}\Big[g|_{\mathbb{S}^{n_1-1} \times ... \times \mathbb{S}^{n_k-1}} (x) \geq t\Big]
\end{align*}
where $n = n_1 + ... + n_k$.
\end{corollary}
For a matrix $A \in \R^{n \times n}$ with rows $a_1,...,a_n$, we define the corresponding preconditioned matrix  $\Tilde{A}$ as the matrix with normalized rows $\frac{a_i}{||a_i||}$. Then we define $\overline{\kappa}_{2\infty}(A) := \kappa_{2\infty}(\Tilde{A})$. We can also similarly define $\overline{\kappa}(A) := \kappa(\Tilde{A})$.
\begin{lemma}\label{lem:normrelations}[Lemma 2.47 in \cite{BC13}]
For $A \in \R^{n \times n} \setminus \{0\}$
\begin{align*}
    \frac{1}{\sqrt{n}}\kappa(A) \leq \kappa_{2\infty}(A) \leq \sqrt{n}\kappa(A).
\end{align*}
\end{lemma}
\begin{corollary}\label{corr:conditionnumber}
 For any $n > 2$ and $0 < \varepsilon \leq 1$,
\begin{equation}
    \text{Pr}\Big[\overline{\kappa}(A) \geq \frac{1}{\varepsilon}\Big] \leq \sqrt{\frac{2}{\pi}}n^2\varepsilon.
\end{equation}
where $A$ is picked at random from $\mathcal{N}(0,I_{n^2})$.   
\end{corollary}
\begin{proof}
Using Theorem \ref{thm:kappabounds} and the fact that for $A \in (\mathbb{S}^{n-1})^n$, $A = \Tilde{A}$, we get that
\begin{equation}\label{eq:kappabar2infty}
   \text{Pr}\Big[\overline{\kappa}_{2\infty}(A) \geq \frac{1}{\varepsilon}\Big] = \text{Pr}\Big[\kappa_{2\infty}(A) \geq \frac{1}{\varepsilon}\Big] \leq \sqrt{\frac{2}{\pi}}n^2\varepsilon
\end{equation}
when $A$ is picked uniformly at random from the uniform distribution on $(\mathbb{S}^{n-1})^n$.
\par
Note that by construction, $\overline{\kappa}_{2\infty}(A)$ is an integrable function that is scale-invariant in each row of $A$. Hence, the conditions of Corollary \ref{corr:conditionnumber} is satisfied for $k = n$ and $n_i = n$ for all $i \in [n]$. Then using Corollary \ref{corr:conditionnumber} along with (\ref{eq:kappabar2infty}), we get that for all $\varepsilon \in \R \setminus \{0\}$,
\begin{equation}\label{eq:normalprob}
    \text{Pr}_{A \sim \mathcal{N}(0,I_{n^2})}\Big[\overline{\kappa}_{2 \infty}(A) \geq \frac{1}{\varepsilon}\Big] \leq \sqrt{\frac{2}{\pi}}n^2\varepsilon.
\end{equation}
Using the relation between $\kappa_{2\infty}$ and $\kappa$ on $\Tilde{A}$ from Lemma \ref{lem:normrelations} for some $A \sim \mathcal{N}(0,I_{n^2})$, we get that $  \frac{1}{\sqrt{n}}\overline{\kappa}(A) \leq \overline{\kappa}_{2\infty}(A)$. Replacing this in (\ref{eq:normalprob}), we get that 
\begin{align*}
    \text{Pr}_{A \sim \mathcal{N}(0,I_{n^2})}\Big[\overline{\kappa}(A) \geq \frac{\sqrt{n}}{\varepsilon}\Big] \leq \sqrt{\frac{2}{\pi}}n^{\frac{5}{2}}\varepsilon.
\end{align*}
This gives us the desired result.
\end{proof}
\begin{lemma}[Laurent-Massart Bounds]\label{lem:chitail}\cite{LM00}
 Let $v \in \mathcal{N}(0,1)^n$ be a Gaussian random vector in $\R^n$. Then for all $\alpha \in (0,\frac{1}{2})$
 \begin{align*}
     \text{Pr}\Big[ ||v||^2 \in [n-2n^{\frac{1}{2} + \alpha},n+2n^{\frac{1}{2} + \alpha} + 2n^{\alpha}]\Big] \geq 1-\frac{1}{e^{n^{\alpha}}}.
 \end{align*}
\end{lemma}

The main theorem of this section is the following:
\begin{theorem}\label{thm:averageproof}
Let $U = (u_{ij})\in \text{GL}_n(\R)$ be a matrix such that for all $i,j \in~[n]$, its $(i,j)$-th entry is sampled independently at random from the standard Gaussian distribution. More formally, $u_{i,j} \sim \mathcal{N}(0,1)$ for all $i,j \in [n]$. Then $$\kappa_F(U) \leq 5n^2 + 2n^7$$ with probability at least $1 - (\sqrt{\frac{2}{n^2 \pi}} + \frac{n}{e^{n^{\frac{1}{4}}}})$ for all $n \geq 256$.
\end{theorem}
\begin{proof}
Let $u_1,...,u_n$ be the columns of $U$ and let us  define $\Tilde{U} := \Big(\frac{u_1}{||u_1||},...,\frac{u_n}{||u_n||}\Big)$. By construction,
$U = \Tilde{U}D$ where $D = \text{diag}(||u_1||,...,||u_n||)$ and  $||\Tilde{U}||_F = \sqrt{n}$. Using the definition of the Frobenius condition number and  the fact that $||U||_F \leq \sqrt{n}||U||$, we can write
\begin{equation}\label{eq:expression}
\begin{split}
    \kappa_F(U) = ||U||_F^2 + ||U^{-1}||_F^2 &\leq ||U||_F^2 + ||(\Tilde{U})^{-1}||_F^2||D^{-1}||^2_F \\
    &= ||U||_F^2 + \frac{1}{n}||\Tilde{U}||_F^2||(\Tilde{U})^{-1}||_F^2||D^{-1}||^2_F \\
    &\leq ||U||_F^2 + n||\Tilde{U}||^2||(\Tilde{U})^{-1}||^2||D^{-1}||^2_F \\
    &= ||U||_F^2 + n(\overline{\kappa}(U))^2\Big(\sum_{i=1}^n \frac{1}{||u_i||^2}\Big).    
\end{split}
\end{equation}
{ Using Lemma \ref{lem:chitail} for $v = u_i$ and $\alpha = \frac{1}{4}$, we get that
\begin{equation}
    \text{Pr}_{u_{ij} \in \mathcal{N}(0,1)}\Big[n - 2n^{\frac{3}{4}} \leq ||u_i||^2 \leq n + 2n^{\frac{3}{4}} + 2n^{\frac{1}{4}}\Big] \geq 1- \frac{1}{e^{n^{\frac{1}{4}}}}.
\end{equation}
Using the fact for all $n \geq 256$, $2n^{\frac{3}{4}} \leq \frac{n}{2}$, we have that for all $i \in [n]$,
\begin{equation}
    \text{Pr}_{u_{ij} \in \mathcal{N}(0,1)}\Big[||u_i||^2 \leq 5n \text{ and } \frac{1}{||u_i||^2} \leq \frac{2}{n}\Big] \geq 1- \frac{1}{e^{n^{\frac{1}{4}}}}.
\end{equation}
Using the union bound, we get that
\begin{equation}\label{eq:normUinv}
    \text{Pr}_{u_{ij} \in \mathcal{N}(0,1)}\Big[||U||_F^2 = \sum_{i =1}^n ||u_i||^2 \leq 5n^2 \text{ and } \sum_{i=1}^n \frac{1}{||u_i||^2} \leq 2\Big] \geq 1- \frac{n}{e^{n^{\frac{1}{4}}}}.
\end{equation}
}
Moreover, by Corollary \ref{corr:conditionnumber}, 
\begin{equation}
    \text{Pr}_{u_{ij} \in \mathcal{N}(0,1)} \Big[\overline{\kappa}(U) \leq n^3\Big] \geq 1 - \sqrt{\frac{2}{\pi}}\frac{1}{n}
\end{equation}
where each row $a_i$ of $A$ is picked independently at random from $\mathcal{N}(0,I_n)$.
Combining this with (\ref{eq:normUinv}) and putting it back in (\ref{eq:expression}), we get that
\begin{align*}
    \text{Pr}_{u_{ij} \in \mathcal{N}(0,1)} \Big[\kappa_F(U) \leq 5n^2 + n.n^6.2 \Big] \geq 1 - (\sqrt{\frac{2}{n^2 \pi}} + \frac{n}{e^{n^{\frac{1}{4}}}}).
\end{align*}
This gives us the desired result.
\end{proof}
Using this and the definition of condition numbers for diagonalisable tensors, we can conclude that if a tensor $T \in \R^n \otimes \R^n \otimes \R^n$ is picked at random from the space of symmetric tensors of rank $n$, then the condition number of the tensor
$\kappa(T)$ 
is at most $\text{poly}(n)${ with high probability.}
\end{document}